\documentclass[twocolumn,english,aps,prd,nofootinbib]{revtex4}
\usepackage{lmodern}
\usepackage[T1]{fontenc}
\usepackage[latin9]{inputenc}
\usepackage{xcolor}
\usepackage{babel}
\usepackage{varioref}
\usepackage{enumitem}
\usepackage{amsmath}
\usepackage{amsthm}
\usepackage{amssymb}
\usepackage{graphicx}
\usepackage[unicode=true]
 {hyperref}

\makeatletter

\newcommand{\lyxdot}{.}

\@ifundefined{textcolor}{}
{%
 \definecolor{BLACK}{gray}{0}
 \definecolor{WHITE}{gray}{1}
 \definecolor{RED}{rgb}{1,0,0}
 \definecolor{GREEN}{rgb}{0,1,0}
 \definecolor{BLUE}{rgb}{0,0,1}
 \definecolor{CYAN}{cmyk}{1,0,0,0}
 \definecolor{MAGENTA}{cmyk}{0,1,0,0}
 \definecolor{YELLOW}{cmyk}{0,0,1,0}
}
\theoremstyle{plain}
\newtheorem{thm}{\protect\theoremname}
\theoremstyle{plain}
\newtheorem{lem}[thm]{\protect\lemmaname}
\ifx\proof\undefined
\newenvironment{proof}[1][\protect\proofname]{\par
	\normalfont\topsep6\p@\@plus6\p@\relax
	\trivlist
	\itemindent\parindent
	\item[\hskip\labelsep\scshape #1]\ignorespaces
}{%
	\endtrivlist\@endpefalse
}
\providecommand{\proofname}{Proof}
\fi
\theoremstyle{plain}
\newtheorem{cor}[thm]{\protect\corollaryname}
\theoremstyle{remark}
\newtheorem{rem}[thm]{\protect\remarkname}
\newlist{casenv}{enumerate}{4}
\setlist[casenv]{leftmargin=*,align=left,widest={iiii}}
\setlist[casenv,1]{label={{\itshape\ \casename} \arabic*.},ref=\arabic*}
\setlist[casenv,2]{label={{\itshape\ \casename} \roman*.},ref=\roman*}
\setlist[casenv,3]{label={{\itshape\ \casename\ \alph*.}},ref=\alph*}
\setlist[casenv,4]{label={{\itshape\ \casename} \arabic*.},ref=\arabic*}
\theoremstyle{plain}
\newtheorem{conjecture}[thm]{\protect\conjecturename}

\usepackage{enumerate}

\makeatother

\providecommand{\casename}{Case}
\providecommand{\conjecturename}{Conjecture}
\providecommand{\corollaryname}{Corollary}
\providecommand{\lemmaname}{Lemma}
\providecommand{\remarkname}{Remark}
\providecommand{\theoremname}{Theorem}

\begin{document}
\title{Beyond Schwarzschild-de Sitter spacetimes: II. $\ $An exact non-Schwarzschild
metric \vskip2pt in pure $R^{2}$ gravity and new anomalous properties
of $R^{2}$ spacetime}
\author{Hoang Ky Nguyen$\,$}
\email[\ \,Email:\ \ ]{hoang.nguyen@ubbcluj.ro}

\affiliation{Department of Physics, Babe\c{s}--Bolyai University, Cluj-Napoca
400084, Romania}
\date{\today}
\begin{abstract}
\vskip2pt In a recent publication \textcolor{purple}{{[}\href{https://journals.aps.org/prd/abstract/10.1103/PhysRevD.106.104004}{Phys. Rev. D 106, 104004 (2022)}}{]},
we advanced a program that Buchdahl originated but prematurely abandoned
circa 1962 \textcolor{purple}{{[}\href{https://link.springer.com/article/10.1007/BF02733549}{Nuovo Cimento, Vol. 23, No 1, 141 (1962)}}{]}.
Therein we obtained an exhaustive class of metrics that constitute
the branch of non-trivial solutions to the pure $\mathcal{R}^{2}$
field equation in vacuo. The Buchdahl-inspired metrics in general
possess \emph{non-constant} scalar curvature, thereby defeating the
generalized Lichnerowicz theorem advocated in \citep{Nelson-2010,Lu-2015-a,Lu-2015-b,Luest-2015-backholes}.
We found that the said theorem makes an overly strong assumption on
the asymptotic falloff in the spatial derivatives of the Ricci scalar,
rendering it violable against the Buchdahl-inspired metrics. In this
paper, we shall further extend our work mentioned above \citep{Nguyen-2022-Buchdahl}
by showing that, within the class of Buchdahl-inspired metrics, the
\emph{asymptotically flat} member takes on the following \emph{exact
closed analytical} expression
\[
ds^{2}=\left|1-\frac{r_{\text{s}}}{r}\right|^{\frac{k}{r_{\text{s}}}}\left\{ -\left(1-\frac{r_{\text{s}}}{r}\right)dt^{2}+\left(1-\frac{r_{\text{s}}}{r}\right)^{-1}\frac{\rho^{4}(r)}{r^{4}}\,dr^{2}+\rho^{2}(r)\,d\Omega^{2}\right\} 
\]
in which the \emph{areal} coordinate $\rho$ is related to the radial
coordinate $r$ per
\[
\rho(r)=\frac{\zeta\,r_{\text{s}}\left|1-\frac{r_{\text{s}}}{r}\right|^{\frac{1}{2}\left(\zeta-1\right)}}{\left|1-\text{sgn}\Bigl(1-\frac{r_{\text{s}}}{r}\Bigr)\left|1-\frac{r_{\text{s}}}{r}\right|^{\zeta}\right|};\ \ \ \ \ \zeta:=\sqrt{1+3k^{2}/r_{\text{s}}^{2}}
\]
The \emph{special} Buchdahl-inspired metric, as we shall call it as
such hereafter, is characterized by a ``Schwarzschild'' radius $r_{\text{s}}$
and the Buchdahl parameter $k$, the latter of which arises via the
higher-derivative nature of $\mathcal{R}^{2}$ gravity. The case $k=0$
corresponds precisely to the classic Schwarzschild metric. Equipped
with this exact expression, we shall investigate pure $\mathcal{R}^{2}$
spacetime structures. The asymptotically flat spacetime is split into
an interior region and an exterior region, with the boundary situated
at $\rho=r_{\text{s}}$. We find that, except for $k=0$ and $k=-r_{\text{s}}$,
the Kretschmann invariant of this metric exhibits an additional singularity
at the interior-exterior boundary. Accordingly, the surface area of
the interior-exterior boundary is found to vanish for $k\in(-\infty,-r_{\text{s}})\cup(0,+\infty)$,
diverge for $k\in(-r_{\text{s}},0)$, equal $4\pi r_{\text{s}}^{2}$
for $k=0$, and equal $16\pi r_{\text{s}}^{2}$ for $k=-r_{\text{s}}$.
This behavior signals a naked singularity or a wormhole. We shall
also analytically construct the Kruskal-Szekeres (KS) diagram for
pure $\mathcal{R}^{2}$ spacetime. The Buchdahl parameter $k$ is
found to modify the KS diagram in some fundamental way. A striking
result is that the (modified) KS diagram develops a ``gulf'' that
sandwiches between the four established quadrants. The ``gulf''
resides strictly on the interior-exterior boundary and does not correspond
to any domain in the physical spacetime, specified by $\left(t,r,\theta,\phi\right)$.
The nature of this novel ``virtual'' region in the KS diagram is
an open question, related to which we make a conjecture on a possible
path forward.
\end{abstract}
\maketitle

\section{\label{sec:Introduction}Introduction: $\ $Buchdahl's $\boldsymbol{1962}$
program in pure $\mathcal{R}^{2}$ gravity}

Pure $\mathcal{R}^{2}$ gravity is among the simplest candidates for
modified gravity. Its action contains a single term, $\frac{1}{2\kappa}\int d^{4}x\sqrt{-g}\,\mathcal{R}^{2}$,
with $\kappa$ being a dimensionless parameter, while the traditional
Einstein-Hilbert term is suppressed. The theory was considered as
early as the 1960's by Buchdahl as a parsimonious prototype of higher-order
gravity that possesses an additional symmetry -- the scale invariance
\citep{Buchdahl-1962}. There is a surge of interest in the pure $\mathcal{R}^{2}$
action of late \citep{AlvarezGaume-2015,Alvarez-2018,Stelle-1977,Edery-2014}
within a larger context of modified gravity \citep{Capozziello-2011,Clifton-2011,deFelice-2010,Sotiriou-2008,Nojiri-2011,Nojiri-2017}.
Pure $\mathcal{R}^{2}$ gravity is the only theory that is both ghost-free
and scale invariant \citep{Stelle-1978,Luest-2015-fluxes}. \vskip4pt

In a seminal -- yet obscure -- 1962 Nuovo Cimento paper entitled
\emph{``On the Gravitational Field Equations Arising from the Square
of the Gaussian Curvature''} \citep{Buchdahl-1962}, Buchdahl pioneered
a program in search of static spherically symmetric vacua for pure
$\mathcal{R}^{2}$ gravity. He established therein that the vacua
in general possess \emph{non-constant} scalar curvature, as a result
of the higher-derivative structure of the theory. Surpassing several
obstacles, his efforts culminated in a non-linear second-order ordinary
differential equation (ODE) \emph{which required being solved}. The
finish line was within \emph{his} striking distance: the $\mathcal{R}^{2}$
vacua Buchdahl sought after hinged on the analytical solution --
yet to be found in his time -- to the ODE he derived. Unfortunately,
Buchdahl deemed his ODE intractable and prematurely suspended his
pursuit for an analytical solution. Until our recent work \citep{Nguyen-2022-Buchdahl},
his ODE had remained untackled; and to this day, his Nuovo Cimento
paper has largely gone unnoticed by the gravitation research community
\footnote{Buchdahl's paper has gathered merely $40+$ citations since its publications
in 1962, according to NASA ADS and InpireHEP citation trackers. Yet,
none of these citations attempted to solve Buchdahl's ODE.}. \vskip4pt

Recently, we have managed to bridge the remaining gap in the Buchdahl
program by identifying a compact solution to his ODE \citep{Nguyen-2022-Buchdahl}.
With this impasse finally overcome, we proceeded to accomplishing
Buchdahl's ultimate goal. The outcome is an exhaustive class of pure
$\mathcal{R}^{2}$ vacua expressible in a compact form, which we called
\emph{the Buchdahl-inspired solution}, to be summarized below. \vskip4pt

\subsubsection*{\textbf{The Buchdahl-inspired solution}}

\begin{figure}[t]
\noindent \begin{centering}
\includegraphics[scale=0.8]{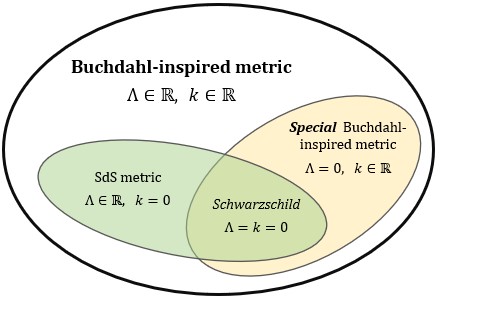}
\par\end{centering}
\noindent \begin{centering}
\includegraphics[scale=0.8]{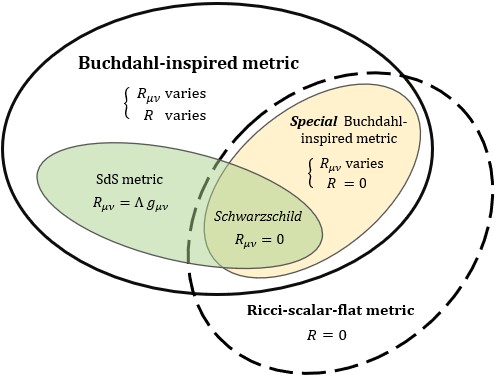}
\par\end{centering}
\caption{\label{fig:Buchdahl-inspired-metric-family}Upper panel: the Buchdahl-inspired
metric family and its subsets. Lower panel: their relation with the
Ricci-scalar-flat family. The \emph{special} Buchdahl-inspired metric
in the intersection is asymptotically flat, whereas the Buchdahl-inspired
metric with $\Lambda\protect\neq0$ is asymptotically constant.}
\end{figure}

In \citep{Nguyen-2022-Buchdahl} by reformulating Buchdahl's original
derivation which was quite cumbersome, we obtained the Buchdahl-inspired
metric, cast in a parallel resemblance to the classic Schwarzschild-de
Sitter (SdS) metric, per\vskip-10pt

\small
\begin{equation}
ds^{2}=e^{k\int\frac{dr}{r\,q(r)}}\left\{ p(r)\Bigl[-\frac{q(r)}{r}dt^{2}+\frac{r}{q(r)}dr^{2}\Bigr]+r^{2}d\Omega^{2}\right\} \label{eq:B-metric-1}
\end{equation}
\normalsize The pair of functions $\{p(r),q(r)\}$ obey the ``evolution''
rules\small
\begin{align}
{\displaystyle \frac{dp}{dr}}\, & ={\displaystyle \,\frac{3k^{2}}{4\,r}\frac{p}{q^{2}}}\label{eq:B-metric-2}\\
{\displaystyle {\displaystyle \frac{dq}{dr}}}\, & =\,{\displaystyle \bigl(1-\Lambda\,r^{2}\bigr)\,p}\label{eq:B-metric-3}
\end{align}
\normalsize and the \emph{non-constant} Ricci scalar equals to
\begin{equation}
\mathcal{R}(r)=4\Lambda\,e^{-k\int\frac{dr}{r\,q(r)}}\label{eq:B-metric-4}
\end{equation}
This metric is specified by two parameters, $\Lambda$ and $k$, resulted
from the fourth-derivative nature of $\mathcal{R}^{2}$ gravity, a
theory that requires two additional boundary conditions as compared
with second-derivative theories, such as the Einstein-Hilbert action.
If the spacetime structures associated with this metric are proven
to be stable, then $k$ would stand for new higher-derivative hair
which allows the Ricci scalar to vary on the manifold, per Eq. \eqref{eq:B-metric-4}.
At largest distances, the Ricci scalar converges to $4\Lambda$, characterizing
an asymptotically constant spacetime.\vskip4pt

To allay any lingering doubt, in \citep{Nguyen-2022-Buchdahl} and
\citep{Shurtleff-2022} the current author and Shurtleff independently
checked that the solution given in Eqs. \eqref{eq:B-metric-1}--\eqref{eq:B-metric-4}
satisfies the pure $\mathcal{R}^{2}$ vacuo field equation

\begin{equation}
\mathcal{R}\Bigl(\mathcal{R}_{\mu\nu}-\frac{1}{4}g_{\mu\nu}\mathcal{R}\Bigr)+\Bigl(g_{\mu\nu}\,\square-\nabla_{\mu}\nabla_{\nu}\Bigr)\mathcal{R}=0\label{eq:field-eqn}
\end{equation}
for all values of $\Lambda\in\mathbb{R}$ and $k\in\mathbb{R}$, thereby
affirming its validity. We must stress that the solution presented
above is able to \emph{defeat} the generalized Lichnerowicz theorem
advocated in \citep{Nelson-2010,Lu-2015-a,Lu-2015-b,Luest-2015-backholes}
by evading an overly strong condition on the asymptotic falloff in
$D_{i}\mathcal{R}$ assumed in the theorem; see our companion papers
in this ``Beyond Schwarzschild--de Sitter spacetimes'' series for
a detailed exposition \citep{Nguyen-2022-Buchdahl,Nguyen-2022-extension}.
\vskip4pt

The most crucial element of the metric is the new (Buchdahl) parameter
$k$ which makes the metric \emph{non-Schwarzschild}. At $k=0$, the
Buchdahl-inspired metric duly recovers the SdS metric. To see this,
at $k=0$ the evolution rules \eqref{eq:B-metric-2}--\eqref{eq:B-metric-3}
admit the solution $p(r)\equiv1$ and $q(r)=r-\frac{\Lambda}{3}r^{3}-r_{\text{s}}$,
with $r_{\text{s}}$ being a constant, upon which metric \eqref{eq:B-metric-1}
is readily brought into the SdS form with a constant curvature $\mathcal{R}=4\Lambda$
everywhere. A non-zero value of $k$ would trigger a non-linear interplay
between $p(r)$ and $q(r)$ per Eqs. \eqref{eq:B-metric-2}--\eqref{eq:B-metric-3}
and enable a \emph{non-constant} curvature to manifest, per Eq. \eqref{eq:B-metric-4}.\vskip4pt

The relations between the Buchdahl-inspired metric and the SdS metric
as well as the null-Ricci-scalar spaces are depicted by the Venn diagrams
in Fig. \ref{fig:Buchdahl-inspired-metric-family}. By \emph{superseding}
the SdS metric, the Buchdahl-inspired spacetime is a bona fide enlargement
of the SdS spacetime, suitably regarded as a framework \emph{``beyond
Schwarzschild--de Sitter}'' \citep{Nguyen-2022-Buchdahl}.

\subsubsection*{\textbf{The curious case of $\Lambda=0$}}

\textcolor{black}{Also shown in Fig. \ref{fig:Buchdahl-inspired-metric-family}
is the }\textcolor{black}{\emph{special}}\textcolor{black}{{} Buchdahl-inspired
metric which is the Buchdahl-inspired metric with $\Lambda$ set equal
to zero. This special metric wholly occupies the intersection of the
branch of (non-trivial) Buchdahl-inspired metrics and the branch of
(trivial) null-Ricci-scalar spaces. \vskip4pt}

\textcolor{black}{Surprisingly, despite being non-linear, the evolution
rules \eqref{eq:B-metric-2} and \eqref{eq:B-metric-3} are }\textcolor{black}{\emph{fully
soluble}}\textcolor{black}{{} for $\Lambda=0$. In this paper we shall
exploit this advantage to derive a }\textcolor{black}{\emph{closed
analytical}}\textcolor{black}{{} expression for the }\textcolor{black}{\emph{special}}\textcolor{black}{{}
Buchdahl-inspired metric.\vskip4pt}

Equipped with this exact analytical solution, we then are empowered
to investigate the properties of $\mathcal{R}^{2}$ spacetime structures
that live on an asymptotically flat background. These structures are
described by the \emph{special} Buchdahl-inspired metric.\vskip4pt
\begin{center}
-----------------$\infty$-----------------
\par\end{center}

\textcolor{black}{Our paper is organized in four major sections. Sec.
\ref{sec:Derivation} is devoted to deriving the }\textcolor{black}{\emph{special}}\textcolor{black}{{}
Buchdahl-inspired metric.} Sec. \ref{sec:Application-I} produces
a number of surprising properties in the Kretschmann invariant and
the surface area of the interior-exterior boundary of $\mathcal{R}^{2}$
spacetime structures. Sec. \ref{sec:Application-II} analytically
constructs a \emph{modified} Kruskal-Szekeres (KS) diagram for the
\emph{special} Buchdahl-inspired metric and uncovers yet a novel feature
of its KS diagram. Finally, Sec. \ref{sec:Summary} discusses the
potential implications of our finding in various areas in modified
gravity.

\section{\label{sec:Derivation}Derivation of the \emph{special} Buchdahl-inspired
metric}

This rather dense section derives the closed analytical solution in
step-by-step details, with Lemma \ref{lem:lem-final} being our ultimate
result. \textcolor{black}{We start with solving the evolution rules
\eqref{eq:B-metric-2}--\eqref{eq:B-metric-3} for $\Lambda=0$ in
Sec. \ref{subsec:Analytical-solution}. We then, in Sec. \ref{subsec:Problems},
expose the inadequacy of the standard Schwarzschild radial coordinate
$r$ for this metric, resulting in the need for a }\textcolor{black}{\emph{new}}\textcolor{black}{{}
radial coordinate. Secs. \ref{subsec:First-change} and \ref{subsec:Second-change}
introduce two coordinate transformations in sequel that lead to the
final solution, described in Sec. \ref{subsec:The-special-Buchdahl-inspired}. }

\subsection{\label{subsec:Analytical-solution}Analytical solution to the evolution
rules with $\Lambda=0$}

\begin{figure*}[!t]
\noindent \begin{centering}
\includegraphics[scale=0.85]{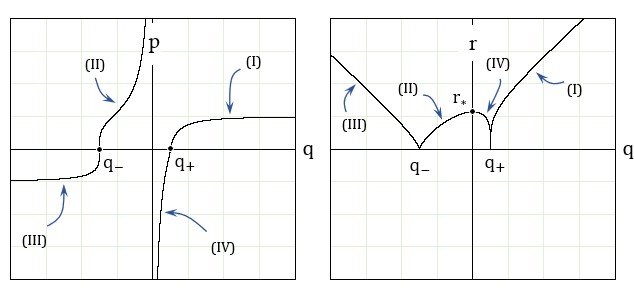}
\par\end{centering}
\noindent \begin{centering}
\includegraphics[scale=0.85]{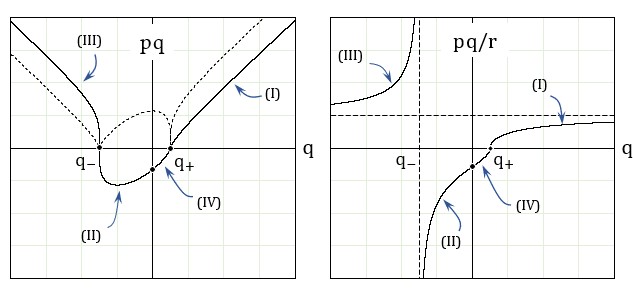}
\par\end{centering}
\caption{\label{fig:Plots-vs-q}Plots of $p$, $r$, $pq$ and $pq/r$ as functions
of $q$. Plots are for $r_{\text{s}}=1,\ k=r_{\text{s}}$. See Remark
\vref{rem:rem-plot-vs-q} for explanations.}
\end{figure*}

\begin{lem}
\noindent \label{lem:lemma-1}For $\Lambda=0$, the set of equations
\eqref{eq:B-metric-2}--\eqref{eq:B-metric-3} admits the following
solution:
\begin{align}
r & =\left|q-q_{+}\right|^{\frac{q_{+}}{q_{+}-q_{-}}}\left|q-q_{-}\right|^{-\frac{q_{-}}{q_{+}-q_{-}}}\label{eq:a-0a}\\
p & =\frac{(q-q_{+})(q-q_{-})}{r\,q}\label{eq:a-0b}\\
q_{\pm} & :=\frac{1}{2}\left(-r_{\text{s}}\pm\sqrt{r_{\text{s}}^{2}+3k^{2}}\right)\label{eq:a-0c}
\end{align}
with $r_{\text{s}}\in\mathbb{R}$ and $q_{\pm}$ representing the
two real roots of the algebraic equation
\begin{equation}
q^{2}+r_{\text{s}}\,q-\frac{3k^{2}}{4}=0\label{eq:a-0d}
\end{equation}
\vskip8pt
\end{lem}
\begin{proof}
\noindent For $\Lambda=0$, the evolution rules \eqref{eq:B-metric-2}--\eqref{eq:B-metric-3}
become
\begin{align}
p_{r} & =\frac{3k^{2}}{4r}\frac{p}{q^{2}}\label{eq:a-1a}\\
q_{r} & =p\label{eq:a-1b}
\end{align}
which give
\begin{equation}
q_{rr}=\frac{3k^{2}}{4r}\frac{q_{r}}{q^{2}}\label{eq:a-1c}
\end{equation}
Upon a change of variable $r=e^{x}$:
\begin{align}
q_{r} & =\frac{dq}{dx}\frac{dx}{dr}=q_{x}\,e^{-x}\label{eq:a-1d}\\
q_{rr} & =\frac{d}{dx}\left(q_{x}e^{-x}\right)\frac{dx}{dr}=\left(q_{xx}-q_{x}\right)e^{-2x}\label{eq:a-1e}
\end{align}
Equation \eqref{eq:a-1c} becomes
\begin{equation}
q_{xx}=\left(1+\frac{3k^{2}}{4q^{2}}\right)q_{x}\label{eq:a-2a}
\end{equation}
which can be recast as
\begin{equation}
\frac{d}{dx}\left(\frac{dq}{dx}\right)=\left(1+\frac{3k^{2}}{4q^{2}}\right)\frac{dq}{dx}\label{eq:a-2b}
\end{equation}
or, equivalently
\begin{equation}
\frac{d}{dq}\left(\frac{dq}{dx}\right)=1+\frac{3k^{2}}{4q^{2}}\label{eq:a-2c}
\end{equation}
Upon integrating, it yields a first-order ODE
\begin{equation}
\frac{dq}{dx}=q-\frac{3k^{2}}{4q}+r_{\text{s}}\label{eq:a-2d}
\end{equation}
with $r_{\text{s}}$ being an integration constant. Let $q_{\pm}:=\frac{1}{2}\left(-r_{\text{s}}\pm\sqrt{r_{\text{s}}^{2}+3k^{2}}\right)$
be the two real roots of the algebraic equation \eqref{eq:a-0d}.
A further integration of \eqref{eq:a-2d}, with the integration constant
for $x$ set equal zero without loss of generality, produces \vskip-2pt

\noindent \small
\begin{align}
x & =\int\frac{q\,dq}{(q-q_{+})(q-q_{-})}\label{eq:a-3a}\\
 & =\frac{q_{+}}{q_{+}-q_{-}}\ln|q-q_{+}|-\frac{q_{-}}{q_{+}-q_{-}}\ln|q-q_{-}|\label{eq:a-3b}
\end{align}
\normalsize Restoring $r=e^{x}$, we then obtain
\begin{equation}
r=\left|q-q_{+}\right|^{\frac{q_{+}}{q_{+}-q_{-}}}\left|q-q_{-}\right|^{-\frac{q_{-}}{q_{+}-q_{-}}}\label{eq:a-3c}
\end{equation}
Additionally, from \eqref{eq:a-1b} and \eqref{eq:a-2d}, together
with $x=\ln r$, we have
\begin{align}
p & =q_{r}=q_{x}\frac{dx}{dr}=q_{x}\,\frac{1}{r}\label{eq:a-3e}\\
 & =\frac{1}{r}\left(q-\frac{3k^{2}}{4q}+r_{\text{s}}\right)\label{eq:a-3f}\\
 & =\frac{1}{rq}(q-q_{+})(q-q_{-})\label{eq:a-3g}
\end{align}
\end{proof}
\noindent Combining Lemma \ref{lem:lemma-1} with Eq. \eqref{eq:B-metric-1},
we arrive at the following analytical result.\vskip12pt
\begin{cor}
\noindent For $\Lambda=0$, the Buchdahl-inspired metric \eqref{eq:B-metric-1}--\eqref{eq:B-metric-3}
is fully analytic, per \vskip0pt

\noindent \small
\begin{align}
ds^{2} & =e^{k\int\frac{dr}{r(q)q}}\biggl\{-\frac{p(q)q}{r(q)}dt^{2}+\frac{p(q)r(q)}{q}dr^{2}+r^{2}(q)d\Omega^{2}\biggr\}\label{eq:a-4a}\\
r(q) & =\left|q-q_{+}\right|^{\frac{q_{+}}{q_{+}-q_{-}}}\left|q-q_{-}\right|^{-\frac{q_{-}}{q_{+}-q_{-}}}\label{eq:a-4b}\\
p(q) & =\frac{(q-q_{+})(q-q_{-})}{r(q)\,q}\label{eq:a-4c}\\
k & =\left(-\frac{4}{3}q_{+}q_{-}\right)^{1/2}\label{eq:a-4d}
\end{align}
\normalsize\vskip8pt
\end{cor}
\begin{rem}
\noindent We shall call the Buchdahl-inspired metric with $\Lambda=0$
the \textbf{\emph{special}}\emph{ Buchdahl-inspired metric}. We shall
also choose a convention of $r_{\text{s}}>0$ in the rest of the paper.
The case of $r_{\text{s}}=0$ is considered in Appendix \ref{sec:Case-rS0}.
\end{rem}
\vskip6pt
\begin{rem}
\label{rem:rem-plot-vs-q}Using Eqs. \eqref{eq:a-4b} and \eqref{eq:a-4c},
we produce the plots of $p$, $r$, $pq$ and $pq/r$ against $q$,
as shown in Fig. \ref{fig:Plots-vs-q}.\linebreak The parameters
are $k=r_{\text{s}}=1$, \textcolor{black}{making $q_{+}=1/2$, $q_{-}=-3/2$,
and $r_{*}:=\left|q_{+}\right|^{\frac{q_{+}}{q_{+}-q_{-}}}\left|q_{-}\right|^{-\frac{q_{-}}{q_{+}-q_{-}}}=(27)^{1/4}/2\approx1.14$.
}In the upper left panel, the four quadrants of the ${p,q}$ diagram
are labeled (I), (II), (III), (IV) counterclockwise, respectively.
In the other three panels, the quadrant labels (as defined in the
$\{p,q\}$ plot) are attached accordingly.
\end{rem}
\vskip6pt
\begin{rem}
\noindent From Eq. \eqref{eq:a-4c}, it is straightforward to prove
that the \emph{special} Buchdahl-inspired metric supports a duality
relation:
\begin{equation}
q\,p(q)=r(q_{+}+q_{-}-q)\label{eq:duality}
\end{equation}
\end{rem}
\vskip6pt
\begin{rem}
Note that $q_{-}<0<q_{+}$, by virtue of their definitions in Eq.
\eqref{eq:a-0c}. The zeros of $r$ and $p$ occur at $q=q_{+}$ and
$q=q_{-}$. Furthermore,
\begin{align}
p & =\begin{cases}
\ >0 & \text{for }q\in(q_{-},0)\cup(q_{+},+\infty)\\
\ <0 & \text{for }q\in(-\infty,q_{-})\cup(0,q_{+})
\end{cases}\\
pq & =\begin{cases}
\ >0 & \text{for }q\in(-\infty,q_{-})\cup(q_{+},+\infty)\\
\ <0 & \text{for }q\in(q_{-},q_{+})
\end{cases}
\end{align}
\end{rem}
\vskip6pt
\begin{rem}
From the duality relation \eqref{eq:duality} and the definition of
$q_{\pm}$ in \eqref{eq:a-0c},
\begin{equation}
\frac{q\,p(q)}{r(q)}=\frac{r(q_{+}+q_{-}-q)}{r(q)}=\left|\frac{q-q_{+}}{q-q_{-}}\right|^{\frac{1}{\sqrt{1+3k^{2}/r_{\text{s}}^{2}}}}
\end{equation}
making $\frac{q\,p}{r}$ vanish as $q\rightarrow q_{+}$ and diverge
as $q\rightarrow q_{-}$. These behaviors account for the lower right
panel in Fig. \ref{fig:Plots-vs-q}.
\end{rem}
\vskip6pt
\begin{rem}
As $q\rightarrow0^{\pm}$, $r$ approaches $r_{*}:=\left|q_{+}\right|^{\frac{q_{+}}{q_{+}-q_{-}}}\left|q_{-}\right|^{-\frac{q_{-}}{q_{+}-q_{-}}}$,
whereas $p\rightarrow\mp\infty$, respectively. One can also show
that, for $q\in(q_{-},q_{+})$,
\begin{equation}
\frac{dr}{dq}=-q\left(q_{+}-q\right)^{\frac{q_{-}}{q_{+}-q_{-}}}\left(q-q_{-}\right)^{-\frac{q_{+}}{q_{+}-q_{-}}}
\end{equation}
forcing $r(q)$ to peak at $q=0$ in the interval $(q_{-},q_{+})$.
These behaviors explain the two upper panels in Fig. \ref{fig:Plots-vs-q}.
\end{rem}

\subsection{\label{subsec:Problems}Problems with the Schwarzschild radial coordinate
in $\mathcal{R}^{2}$ gravity}
\begin{center}
\begin{figure*}[!t]
\begin{centering}
\includegraphics[scale=0.68]{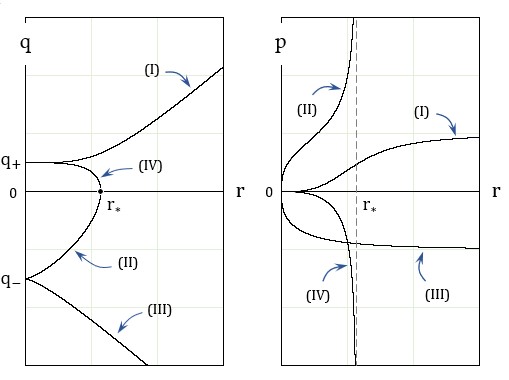}$\ $\includegraphics[scale=0.68]{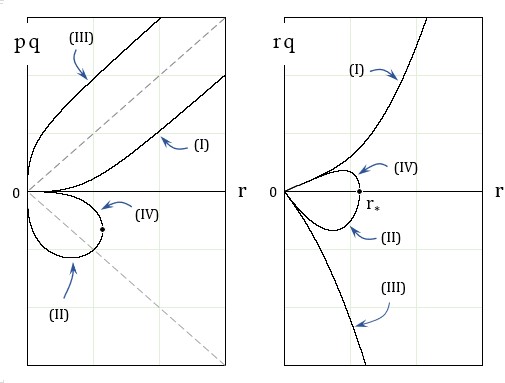}
\par\end{centering}
\caption{\label{fig:Plots-vs-r}Plots of $q$, $p$, $pq$ and $rq$ as functions
of $r$. Plots are for $r_{\text{s}}=1,\ k=r_{\text{s}}$. The leftmost
panel reveals the multi-valuedness problem for $q(r)$.}
\end{figure*}
\par\end{center}

\textcolor{black}{The generic Buchdahl-inspired metric \eqref{eq:B-metric-1}
is expressed in terms of the Schwarzschild coordinate system, $(t,r,\theta,\phi).$
This system would be problematic for metric \eqref{eq:a-4a}--\eqref{eq:a-4d}
however, as we shall see below.\vskip4pt}

\textcolor{black}{Despite Lemma \ref{lem:lemma-1} yielding the relation
$r(q)$, the inversion operation to express $q$ in terms of $r$
using elementary functions cannot be carried out. The reason is that
the two exponents, $\frac{q_{+}}{q_{+}-q_{-}}$ and $\frac{q_{-}}{q_{+}-q_{-}}$,
in \eqref{eq:a-0a} are ``out of sync'' with each other. This trouble
is further complicated by the multi-valuedness of $q(r)$.\vskip4pt}

\textcolor{black}{To see the multi-valuedness problem, we shall re-plot
Fig. \ref{fig:Plots-vs-q} but with a small twist; we shall re-plot
it against the variable $r$ in place of $q$. In Fig. \ref{fig:Plots-vs-r}
we plot $q$, $p$, $pq$ and $rq$ as functions of $r$; again, with
$k=r_{\text{s}}=1$, $q_{+}=1/2$, $q_{-}=-3/2$, and $r_{*}:=\left|q_{+}\right|^{\frac{q_{+}}{q_{+}-q_{-}}}\left|q_{-}\right|^{-\frac{q_{-}}{q_{+}-q_{-}}}=(27)^{1/4}/2\approx1.14$.
The quadrant labels (I), (II), (III) and (IV) defined from Fig. \ref{fig:Plots-vs-q}
are carried over to Fig. \ref{fig:Plots-vs-r}; see Remark \ref{rem:rem-plot-vs-q}.
In the leftmost panel of Fig. \ref{fig:Plots-vs-r}, the function
$q(r)$ is double-valued for $r>r_{*}$, and quadruple-valued for
$r<r_{*}$. This is the multi-valuedness problem which further handicaps
the inversion of $q$ in terms of $r$.\vskip4pt}

The multi-valuedness means that $r$, despite playing \emph{the} Schwarzschild
radial coordinate in metric \eqref{eq:B-metric-1}--\eqref{eq:B-metric-3},
is \emph{not} a suitable variable for metric \eqref{eq:a-4a}--\eqref{eq:a-4c}.
However, looking back at Fig. \ref{fig:Plots-vs-q}, we immediately
realize that the variable $q$ \emph{can} be a suitable coordinate
because all other variables -- viz. $r$, $p$ and others -- are
\emph{single-valued} functions of $q$. This observation guides us
to the first change of variable in the next section.

\subsection{\label{subsec:First-change}A first change of variable}
\begin{cor}
\noindent \label{cor:first-transf}The special Buchdahl-inspired metric
is fully analytic with respect to the variable $q$, per
\end{cor}
\begin{align}
ds^{2} & =e^{\omega(q)}\biggl\{-\frac{p(q)\,q}{r(q)}dt^{2}+\frac{r(q)}{p(q)\,q}dq^{2}+r^{2}(q)d\Omega^{2}\biggr\}\label{eq:a-5a}\\
r(q) & =\left|q-q_{+}\right|^{\frac{q_{+}}{q_{+}-q_{-}}}\left|q-q_{-}\right|^{\frac{q_{-}}{q_{-}-q_{+}}}\label{eq:a-5b}\\
\frac{q\,p(q)}{r(q)} & =\text{sgn}\left(\frac{q-q_{+}}{q-q_{-}}\right)\left|\frac{q-q_{+}}{q-q_{-}}\right|^{\frac{r_{\text{s}}}{q_{+}-q_{-}}}\label{eq:a-5c}\\
e^{\omega(q)} & =\left|\frac{q-q_{+}}{q-q_{-}}\right|^{\frac{k}{q_{+}-q_{-}}}\label{eq:a-5d}\\
k & =\left(-\frac{4}{3}q_{+}q_{-}\right)^{1/2}\label{eq:a-5e}
\end{align}

\begin{proof}
Eq. \eqref{eq:a-1b} gives
\begin{equation}
dr=\frac{dq}{p}\label{eq:a-6b}
\end{equation}
from which, we deduce that
\begin{equation}
\frac{pr}{q}\,dr^{2}=\frac{r}{p\,q}\,dq^{2}\label{eq:a-6c}
\end{equation}
Metric \eqref{eq:a-4a} thus can be brought into \eqref{eq:a-5a},
with the conformal factor
\begin{equation}
e^{\omega(q)}:=e^{k\int\frac{dr}{r\,q(r)}}=e^{k\int\frac{dq}{p(q)r(q)q}}\label{eq:a-6d}
\end{equation}
which by combining with Eqs. \eqref{eq:a-4c}, produces Eq. \eqref{eq:a-5d},
per
\begin{equation}
e^{\omega(q)}=e^{k\int\frac{dq}{(q-q_{+})(q-q_{-})}}=\left|\frac{q-q_{+}}{q-q_{-}}\right|^{\frac{k}{q_{+}-q_{-}}}\label{eq:a-6e}
\end{equation}
Also from Eq. \eqref{eq:a-4c}
\begin{equation}
\frac{q\,p}{r}=\frac{(q-q_{+})(q-q_{-})}{r^{2}}\label{eq:a-6f}
\end{equation}
and, by using Eq. \eqref{eq:a-4b} and noting that $q_{+}+q_{-}=r_{\text{s}}$
by virtue of \eqref{eq:a-0c}, we arrive at Eq. \eqref{eq:a-5c}.
\end{proof}
\vskip8pt
\begin{rem}
\label{rem:An-immediate-improvement}An immediate improvement of metric
\eqref{eq:a-5a} over metric \eqref{eq:a-4a} is that, apart from
the conformal factor, the two components $g_{00}$ and $g_{11}$ are
reciprocal of each other. This feature resembles that in the Schwarzschild
metric.
\end{rem}
\vskip6pt

\subsection{\label{subsec:Second-change}A second change of variable}

Despite getting a step closer to the form of a Schwarzschild metric
(see Remark \ref{rem:An-immediate-improvement} above), the term $pq/r$
in metric \eqref{eq:a-5a} is still rather cumbersome; see Eq. \eqref{eq:a-5c}.
It is thus desirable to find a more transparent alternative to the
coordinate $q$. The lower right panel in Fig. \vref{fig:Plots-vs-q}
suggests a further improvement. Not only is the combination $pq/r$
a single-valued function of $q$, the reverse is also true: $q$ is
a single-valued function of $pq/r$. We shall thus choose $pq/r$
as \emph{the} radial coordinate in replacement of $q$.\vskip4pt

That is to say, let us define a \emph{new} radial coordinate $\rho\in\mathbb{R}$
such that
\begin{equation}
1-\frac{r_{\text{s}}}{\rho}:=\frac{p(q)\,q}{r(q)}\label{eq:a-7a}
\end{equation}
which, by way of \eqref{eq:a-5c}, becomes
\begin{equation}
1-\frac{r_{\text{s}}}{\rho}=\text{sgn}\left(\frac{q-q_{+}}{q-q_{-}}\right)\left|\frac{q-q_{+}}{q-q_{-}}\right|^{\frac{r_{\text{s}}}{q_{+}-q_{-}}}\label{eq:a-7b}
\end{equation}
Remarkably, despite that $q$ is \emph{not} analytically expressible
in terms of $r$ -- a serious hindrance that we alluded to at the
beginning of Sec. \ref{subsec:Problems} -- the relation \eqref{eq:a-7b}
\emph{can} be inverted to express $q$ as a analytical function of
$\rho$. Furthermore, since $r$ is an analytical function of $\rho$
per \eqref{eq:a-5b}, $r$ in turn can be made an analytical function
of $\rho$. The inversion of Eq. \eqref{eq:a-7b} shall be carried
out in the following Lemma.\vskip8pt
\begin{figure*}[!t]
\noindent \begin{centering}
\includegraphics[scale=0.86]{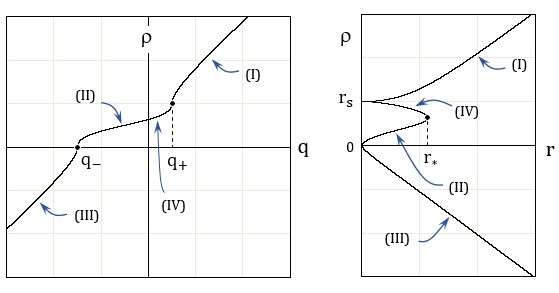}
\par\end{centering}
\caption{\label{fig:rho-vs-q}$\rho$ as functions of $q$ and $r$. Plots
are for $r_{\text{s}}=1,\ k=r_{\text{s}}$.}
\end{figure*}

\noindent \begin{center}
\begin{figure*}[!t]
\noindent \begin{centering}
\includegraphics[scale=0.85]{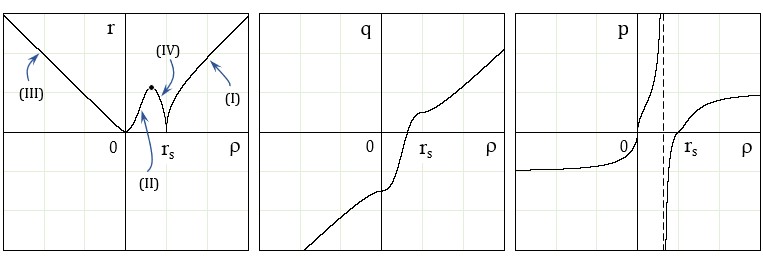}
\par\end{centering}
\caption{\label{fig:r-q-p-vs-rho}Various variables as functions of $\rho$.
Plots are for $r_{\text{s}}=1,\ k=r_{\text{s}}$.}
\end{figure*}
\par\end{center}
\begin{lem}
The Schwarzschild coordinate $r$ is expressible in terms of the variable
$\rho$, per 
\begin{align}
r(\rho) & =\frac{\zeta\,r_{\text{s}}\left|1-\frac{r_{\text{s}}}{\rho}\right|^{\frac{1}{2}\left(\zeta-1\right)}}{\left|1-\text{\emph{sgn}}\left(1-\frac{r_{\text{s}}}{\rho}\right)\left|1-\frac{r_{\text{s}}}{\rho}\right|^{\zeta}\right|}\label{eq:a-8a}\\
\zeta & :=\sqrt{1+3k^{2}/r_{\text{s}}^{2}}\label{eq:a-8b}
\end{align}
\vskip8pt
\end{lem}
\begin{proof}
Denote:
\begin{equation}
x:=1-\frac{r_{\text{s}}}{\rho}\label{eq:a-10c}
\end{equation}
then, from \eqref{eq:a-7b}
\begin{equation}
x=\text{sgn}\frac{q-q_{+}}{q-q_{-}}\left|\frac{q-q_{+}}{q-q_{-}}\right|^{\frac{r_{\text{s}}}{q_{+}-q_{-}}}\label{eq:a-10d}
\end{equation}
Further define $\zeta:=\sqrt{1+3k^{2}/r_{\text{s}}^{2}}\geqslant1\ \ \forall k\in\mathbb{R}$,
then from the definition of $q_{\pm}$ in \eqref{eq:a-0c} we get
\begin{equation}
\frac{r_{\text{s}}}{q_{+}-q_{-}}=\frac{r_{\text{s}}}{\sqrt{r_{\text{s}}^{2}+3k^{2}}}=\frac{1}{\zeta}\label{eq:a-10e}
\end{equation}

\noindent \emph{\vskip12pt Case 1:} $\ $For $q>q_{+}$ then $0<x<1$.\vskip12pt

\noindent Inverting Eq. \eqref{eq:a-10d}:
\begin{align}
\left(\frac{q-q_{+}}{q-q_{-}}\right)^{1/\zeta} & =x=\left|x\right|\label{eq:a-11a}\\
q & =\frac{1}{1-\left|x\right|^{\zeta}}q_{+}-\frac{\left|x\right|^{\zeta}}{1-\left|x\right|^{\zeta}}q_{-}\label{eq:a-11b}
\end{align}
then
\begin{align}
q-q_{+} & =\frac{\left|x\right|^{\zeta}}{1-\left|x\right|^{\zeta}}(q_{+}-q_{-})\label{eq:a-11c}\\
q-q_{-} & =\frac{1}{1-\left|x\right|^{\zeta}}(q_{+}-q_{-})\label{eq:a-11d}
\end{align}
and\small
\begin{align}
r & =(q-q_{+})^{\frac{q_{+}}{q_{+}-q_{-}}}(q-q_{-})^{-\frac{q_{-}}{q_{+}-q_{-}}}\label{eq:a-11e}\\
 & =\left(\frac{\left|x\right|^{\zeta}}{1-\left|x\right|^{\zeta}}\right)^{\frac{q_{+}}{q_{+}-q_{-}}}\left(\frac{1}{1-\left|x\right|^{\zeta}}\right)^{-\frac{q_{-}}{q_{+}-q_{-}}}(q_{+}-q_{-})\label{eq:a-11f}\\
 & =\frac{\left|x\right|^{\zeta\frac{1}{\sqrt{r_{\text{s}}^{2}+3k^{2}}}\frac{1}{2}\left(-r_{\text{s}}+\sqrt{r_{\text{s}}^{2}+3k^{2}}\right)}}{1-\left|x\right|^{\zeta}}\,\zeta r_{\text{s}}\label{eq:a-11g}\\
 & =\zeta r_{\text{s}}\frac{\left|x\right|^{\frac{1}{2}(\zeta-1)}}{1-\left|x\right|^{\zeta}}\label{eq:a-11h}
\end{align}
\normalsize\emph{\vskip12pt Case 2:} $\ $For $q<q_{-}$ then $x>1$.\vskip2pt

\noindent Inverting Eq. \eqref{eq:a-10d}:
\begin{align}
\left(\frac{q-q_{+}}{q-q_{-}}\right)^{1/\zeta} & =x=\left|x\right|\label{eq:a-12a}\\
q & =\frac{1}{1-\left|x\right|^{\zeta}}q_{+}-\frac{\left|x\right|^{\zeta}}{1-\left|x\right|^{\zeta}}q_{-}\label{eq:a-12b}
\end{align}
then
\begin{align}
q-q_{+} & =\frac{\left|x\right|^{\zeta}}{1-\left|x\right|^{\zeta}}(q_{+}-q_{-})\label{eq:a-12c}\\
q-q_{-} & =\frac{1}{1-\left|x\right|^{\zeta}}(q_{+}-q_{-})\label{eq:a-12d}
\end{align}
and\small
\begin{align}
r & =(q_{+}-q)^{\frac{q_{+}}{q_{+}-q_{-}}}(q_{-}-q)^{-\frac{q_{-}}{q_{+}-q_{-}}}\label{eq:a-12e}\\
 & =\left(\frac{\left|x\right|^{\zeta}}{\left|1-\left|x\right|^{\zeta}\right|}\right)^{\frac{q_{+}}{q_{+}-q_{-}}}\left(\frac{1}{\left|1-\left|x\right|^{\zeta}\right|}\right)^{-\frac{q_{-}}{q_{+}-q_{-}}}(q_{+}-q_{-})\label{eq:a-12f}\\
 & =\frac{\left|x\right|^{\zeta\frac{1}{\sqrt{r_{\text{s}}^{2}+3k^{2}}}\frac{1}{2}\left(-r_{\text{s}}+\sqrt{r_{\text{s}}^{2}+3k^{2}}\right)}}{\left|1-\left|x\right|^{\zeta}\right|}\,\zeta r_{\text{s}}\label{eq:a-12g}\\
 & =\zeta r_{\text{s}}\frac{\left|x\right|^{\frac{1}{2}(\zeta-1)}}{\left|1-\left|x\right|^{\zeta}\right|}\label{eq:a-12h}
\end{align}
\normalsize\emph{\vskip12pt Case 3:} $\ $For $q_{-}<q<q_{+}$ then
$x<0$. \vskip12pt

\noindent Inverting Eq. \eqref{eq:a-10d}:
\begin{align}
\left(\frac{q-q_{+}}{q_{-}-q}\right)^{1/\zeta} & =-x=\left|x\right|\label{eq:a-13a}\\
q & =\frac{1}{1+\left|x\right|^{\zeta}}q_{+}+\frac{\left|x\right|^{\zeta}}{1+\left|x\right|^{\zeta}}q_{-}\label{eq:a-13b}
\end{align}
then
\begin{align}
q-q_{+} & =-\frac{\left|x\right|^{\zeta}}{1+\left|x\right|^{\zeta}}(q_{+}-q_{-})\label{eq:a-13c}\\
q-q_{-} & =\frac{1}{1+\left|x\right|^{\zeta}}(q_{+}-q_{-})\label{eq:a-13d}
\end{align}
and\small
\begin{align}
r & =(q_{+}-q)^{\frac{q_{+}}{q_{+}-q_{-}}}(q-q_{-})^{-\frac{q_{-}}{q_{+}-q_{-}}}\label{eq:a-13e}\\
 & =\left(\frac{\left|x\right|^{\zeta}}{1+\left|x\right|^{\zeta}}\right)^{\frac{q_{+}}{q_{+}-q_{-}}}\left(\frac{1}{1+\left|x\right|^{\zeta}}\right)^{-\frac{q_{-}}{q_{+}-q_{-}}}(q_{+}-q_{-})\label{eq:a-13f}\\
 & =\frac{\left|x\right|^{\zeta\frac{1}{\sqrt{r_{\text{s}}^{2}+3k^{2}}}\frac{1}{2}\left(-r_{\text{s}}+\sqrt{r_{\text{s}}^{2}+3k^{2}}\right)}}{1+\left|x\right|^{\zeta}}\,\zeta r_{\text{s}}\label{eq:a-13g}\\
 & =\zeta r_{\text{s}}\frac{\left|x\right|^{\frac{1}{2}(\zeta-1)}}{1+\left|x\right|^{\zeta}}\label{eq:a-13h}
\end{align}
\normalsize

\noindent In all cases, we have
\begin{equation}
r=\zeta\,r_{\text{s}}\frac{\left|x\right|^{\frac{1}{2}(\zeta-1)}}{\left|1-\text{sgn}(x)\left|x\right|^{\zeta}\right|}\label{eq:a-14a}
\end{equation}
which is the desired result, Eq. \eqref{eq:a-8a}.
\end{proof}
\noindent \vskip8pt
\begin{rem}
For illustration, in Fig. \ref{fig:rho-vs-q}, we plot the new variable
$\rho$ against $q$ and $r$. In Fig. \ref{fig:r-q-p-vs-rho}, $r$,
$q$ and $p$ are plotted against $\rho$. In these figures, $k=r_{\text{s}}=1$.
The quadrant labels are attached accordingly.
\end{rem}

\subsection{\label{subsec:The-special-Buchdahl-inspired}The \emph{special} Buchdahl-inspired
metric}

We are now ready for the final step of our derivation. The Buchdahl-inspired
metric with $\Lambda=0$ is provided in Lemma \ref{lem:lem-final}
below.\vskip12pt
\begin{lem}
\noindent \label{lem:lem-final}The \emph{special} Buchdahl-inspired
metric is characterized by 2 parameters, $r_{\text{s}}$ and $\tilde{k}$:\small
\begin{align}
ds^{2} & =\left|1-\frac{r_{\text{s}}}{\rho}\right|^{\tilde{k}}\Biggl\{-\left(1-\frac{r_{\text{s}}}{\rho}\right)dt^{2}+\left(1-\frac{r_{\text{s}}}{\rho}\right)^{-1}\frac{r^{4}(\rho)}{\rho^{4}}\,d\rho^{2}\nonumber \\
 & \ \ \ \ \ \ \ \ \ \ \ \ \ \ \ +r^{2}(\rho)\left(d\theta^{2}+\sin^{2}\theta d\phi^{2}\right)\Biggr\}\label{eq:special-B-1}
\end{align}
\normalsize in which the Schwarzschild radial coordinate $r$ is
related to the new radial coordinate $\rho$ per
\begin{align}
r(\rho) & :=\frac{\zeta\,r_{\text{s}}\left|1-\frac{r_{\text{s}}}{\rho}\right|^{\frac{1}{2}(\zeta-1)}}{\left|1-\text{\emph{sgn}}\left(1-\frac{r_{\text{s}}}{\rho}\right)\left|1-\frac{r_{\text{s}}}{\rho}\right|^{\zeta}\right|}\label{eq:special-B-2}\\
\zeta & :=\sqrt{1+3\tilde{k}^{2}}\label{eq:special-B-3}
\end{align}
\end{lem}
\vskip4pt
\begin{proof}
Firstly, Eq. \eqref{eq:a-7b} leads to

\begin{equation}
\left|1-\frac{r_{\text{s}}}{\rho}\right|=\left|\frac{q-q_{+}}{q-q_{-}}\right|^{\frac{r_{\text{s}}}{q_{+}-q_{-}}}\label{eq:a-9a}
\end{equation}
which neatly brings the conformal factor \eqref{eq:a-5d} to
\begin{equation}
e^{\omega(\rho)}=\left|1-\frac{r_{\text{s}}}{\rho}\right|^{\frac{k}{r_{\text{s}}}}\label{eq:a-9b}
\end{equation}
Secondly, Eq. \eqref{eq:a-9a} is equivalent to
\begin{equation}
\ln\left|1-\frac{r_{\text{s}}}{\rho}\right|=\frac{r_{\text{s}}}{q_{+}-q_{-}}\ln\left|\frac{q-q_{+}}{q-q_{-}}\right|\label{eq:a-9d}
\end{equation}
Taking derivative on both sides of this equation:
\begin{equation}
\left(1-\frac{r_{\text{s}}}{\rho}\right)^{-1}\frac{r_{\text{s}}}{\rho^{2}}\,d\rho=\frac{r_{\text{s}}}{(q-q_{+})(q-q_{-})}\,dq\label{eq:a-9e}
\end{equation}
which, with the aid of Eqs. \eqref{eq:a-5b}, \eqref{eq:a-9a} and
$r_{\text{s}}=-\left(q_{+}+q_{-}\right)$ per \eqref{eq:a-0c}, yields\small
\begin{align}
dq^{2} & =\left|\frac{q-q_{+}}{q-q_{-}}\right|^{-\frac{2r_{\text{s}}}{q_{+}-q_{-}}}(q-q_{+})^{2}(q-q_{-})^{2}\frac{d\rho^{2}}{\rho^{4}}\label{eq:a-9f}\\
 & =\left|q-q_{+}\right|^{\frac{4q_{+}}{q_{+}-q_{-}}}\left|q-q_{-}\right|^{-\frac{4q_{-}}{q_{+}-q_{-}}}\frac{d\rho^{2}}{\rho^{4}}\label{eq:a-9g}\\
 & =r^{4}(q)\frac{d\rho^{2}}{\rho^{4}}\label{eq:a-9h}
\end{align}
\normalsize Finally, by defining $\tilde{k}:=k/r_{\text{s}}$ and
using \eqref{eq:a-7a} and \eqref{eq:a-9h}, the components in metric
\eqref{eq:a-5a} become\small
\begin{align}
g_{tt} & =-e^{\omega}\frac{p\,q}{r}=-\left|1-\frac{r_{\text{s}}}{\rho}\right|^{\tilde{k}}\left(1-\frac{r_{\text{s}}}{\rho}\right)\label{eq:a-9-i}\\
g_{\rho\rho} & =g_{qq}\frac{dq^{2}}{d\rho^{2}}=e^{\omega}\frac{r}{p\,q}\frac{r^{4}}{\rho^{4}}=\left|1-\frac{r_{\text{s}}}{\rho}\right|^{\tilde{k}}\frac{1}{1-\frac{r_{\text{s}}}{\rho}}\frac{r^{4}}{\rho^{4}}\\
g_{\theta\theta} & =e^{\omega}\,r^{2}=\left|1-\frac{r_{\text{s}}}{\rho}\right|^{\tilde{k}}r^{2}\\
g_{\phi\phi} & =g_{\theta\theta}\,\sin^{2}\theta
\end{align}
\normalsize
\end{proof}
\vskip6pt
\begin{rem}
The \emph{rescaled} Buchdahl parameter
\begin{equation}
\tilde{k}:=\frac{k}{r_{\text{s}}}\label{eq:def-tilde-k}
\end{equation}
is a \emph{dimensionless} ratio.
\end{rem}
\vskip6pt
\begin{rem}
At $\tilde{k}=0$, Eqs. \eqref{eq:special-B-2} and \eqref{eq:special-B-3}
yield $\zeta=1$ and $r(\rho)\equiv\rho$. The recovery of the Schwarzschild
metric from metric \eqref{eq:special-B-1} is obvious.
\end{rem}
\vskip6pt
\begin{rem}
The combination $1-\frac{r_{\text{s}}}{\rho}$ is universal in the
\emph{special} Buchdahl-inspired metric, \eqref{eq:special-B-1}--\eqref{eq:special-B-3},
as it is in the classic Schwarzschild metric. The $g_{tt}$ component
flips sign when $\rho$ varies across $r_{\text{s}}$. The radius
$r_{\text{s}}$ plays the role of the ``Schwarzschild'' radius for
pure $\mathcal{R}^{2}$ spacetime structures. 
\end{rem}
\vskip6pt
\begin{rem}
\textcolor{black}{In metric }\eqref{eq:special-B-1}--\eqref{eq:special-B-3}\textcolor{black}{,
the radial coordinate is $\rho$ and the physical ``origin'' is
located at $\rho=0$. The usual Schwarzschild coordinate $r(\rho)$
is not the radial coordinate for this metric. Rather, apart from the
conformal factor $\left|1-\frac{r_{\text{s}}}{\rho}\right|^{\tilde{k}}$,
it acts as an }\textcolor{black}{\emph{areal}}\textcolor{black}{{} coordinate
via the term $r^{2}(\rho)\left[d\theta^{2}+\sin^{2}\theta d\phi^{2}\right]$
in \eqref{eq:special-B-1}.}
\end{rem}
\vskip6pt
\begin{rem}
As $\rho\rightarrow\infty$, per \eqref{eq:special-B-2} we have $r(\rho)\simeq\rho$.
Metric \eqref{eq:special-B-1} asymptotically is \small
\begin{equation}
\left|1-\frac{r_{\text{s}}}{\rho}\right|^{\tilde{k}}\biggl\{-\Bigl(1-\frac{r_{\text{s}}}{\rho}\Bigr)dt^{2}+\frac{r^{4}(\rho)\,d\rho^{2}}{\rho^{4}\Bigl(1-\frac{r_{\text{s}}}{\rho}\Bigr)}+r^{2}(\rho)d\Omega^{2}\biggr\}\label{eq:asymp-Schwa}
\end{equation}
\normalsize We thus do \emph{not} obtain a Schwarzschild spacetime
but a \emph{conformally} Schwarzschild spacetime, with the conformal
factor being $\left|1-\frac{r_{\text{s}}}{\rho}\right|^{\tilde{k}}$.
In principle, the effects of $\tilde{k}$ should manifest via its
influence on the orbital motion of the massive objects, though not
that of light.
\end{rem}
\vskip6pt
\begin{rem}
\textcolor{black}{In \eqref{eq:asymp-Schwa}, since $\left|1-\frac{r_{\text{s}}}{\rho}\right|^{\tilde{k}}\rightarrow1$
when $r\rightarrow\infty$, the }\textcolor{black}{\emph{special}}\textcolor{black}{{}
Buchdahl-inspired metric is asymptotically flat. It can also be verified
to be Ricci-scalar-flat but not Ricci flat. Its 4 non-vanishing Ricci
tensor components are: \small
\begin{align}
\mathcal{R}_{tt} & =\frac{\tilde{k}(\tilde{k}+1)}{2\zeta^{4}}\left|x\right|^{2-2\zeta}\left(1-\text{sgn}(x)\left|x\right|^{\zeta}\right)^{4}\label{eq:a-18a}\\
\mathcal{R}_{\rho\rho} & =\frac{\tilde{k}}{2\rho^{4}\left|x\right|^{2}}\left[3\tilde{k}-1+2\zeta\frac{1+\text{sgn}(x)\left|x\right|^{\zeta}}{1-\text{sgn}(x)\left|x\right|^{\zeta}}\right]\label{eq:a-18b}\\
\mathcal{R}_{\theta\theta} & =\frac{\tilde{k}}{2\zeta^{2}}\left(1-\text{sgn}(x)\left|x\right|^{-\zeta}\right)\times\nonumber \\
 & \left[(k-1)\left(1-\text{sgn}(x)\left|x\right|^{\zeta}\right)+\zeta\left(1+\text{sgn}(x)\left|x\right|^{\zeta}\right)\right]\label{eq:a-18c}\\
\mathcal{R}_{\phi\phi} & =\mathcal{R}_{\theta\theta}\sin^{2}\theta\label{eq:a-18d}
\end{align}
\normalsize in which $x:=1-\frac{r_{\text{s}}}{\rho}$, $\zeta:=\sqrt{1+3\tilde{k}^{2}}$.}
\end{rem}
\vskip12pt

In closing of this main section, Lemma \ref{lem:lem-final} is the
central result of our paper. With this exact analytical result, we
are now well equipped to study the interior-exterior boundary and
the causal structure of $\mathcal{R}^{2}$ spacetime structures in
the rest of our paper.

\section{\label{sec:Application-I}Application I: $\ $Anomalous behavior
of interior/exterior boundary in $\mathcal{R}^{2}$ spacetime}

This section explores and reports a number of novel surprising properties
of the interior-exterior boundary of $\mathcal{R}^{2}$ spacetime,
described by the \emph{special} Buchdahl-inspired metric attained
in Lemma \ref{lem:lem-final}.\vskip4pt

Metric \eqref{eq:special-B-1}--\eqref{eq:special-B-3} bears an
interesting resemblance to the Schwarzschild metric, with four departures:
\begin{itemize}
\item The conformal factor, $\left|1-\frac{r_{\text{s}}}{\rho}\right|^{\tilde{k}}$;
\item The $g_{\rho\rho}$ component contains the ratio $\frac{r^{4}(\rho)}{\rho^{4}}$;
\item The angular part involves $r^{2}(\rho)$ instead of $\rho^{2}$.
\item The function $r(\rho)$ involves the signum function $\text{sgn}\left(1-\frac{r_{\text{s}}}{\rho}\right)$,
thus comprising two distinct expressions, one for $\rho<r_{\text{s}}$
and other for $\rho>r_{\text{s}}$.
\end{itemize}
\vskip1pt

Across $\rho=r_{\text{s}}$, the components $g_{00}$ and $g_{11}$
flip their signs, hence indicating an ``exterior'' region for $\rho\in(r_{\text{s}},+\infty)$
and an ``interior'' region for $\rho\in(0,r_{\text{s}})$. The nature
of the interior-exterior boundary can be deduced from the $\zeta-$Kruskal-Szekeres
diagram constructed in Section \ref{subsec:KS-diagram}. In Fig. \vref{fig:KS-diagram},
the boundary for $\tilde{k}\neq0$ (viz., $\zeta>1$) is the four
hyperbolic branches surrounding Region (VI); the $\rho=r_{\text{s}}$
boundary is \emph{not} a null surface in this situation. For $\tilde{k}=0$,
the four hyperbolic branches degenerate into two straight lines $T=\pm X$
that are null surfaces, making the $\rho=r_{\text{s}}$ boundary the
usual Schwarzschild horizon. For all values of $\tilde{k}\in\mathbb{R}$,
Regions (II) and (IV) in Fig. \ref{fig:KS-diagram} represent ``interior''
sections of an $\mathcal{R}^{2}$ spacetime. \vskip4pt

\subsection{\label{subsec:Areal-coordinate}Behavior of the areal radial coordinate
in pure $\mathcal{R}^{2}$ gravity}

We first start with the \emph{areal} coordinate $r(\rho)$ as a function
of the new radial coordinate $\rho$. The relation is given in Eqs.
\eqref{eq:special-B-2}--\eqref{eq:special-B-3}. The plot of $r(\rho)$
is shown in Fig. \ref{fig:r-vs-rho} for various values of $\tilde{k}$.
In each panel, the curve is juxtaposed against the benchmark $r(\rho)=\rho$
diagonal (dotted) line which corresponds to the case $\tilde{k}=0$
(viz. $\zeta=\sqrt{1+3\tilde{k}^{2}}=1$).\vskip12pt

\noindent \emph{The asymptotics:\vskip12pt}
\begin{itemize}
\item As $\rho\rightarrow0$,
\begin{equation}
r(\rho)\simeq\zeta r_{\text{s}}^{\frac{1}{2}(1-\zeta)}\rho^{\frac{1}{2}(\zeta+1)}\rightarrow0\ \ \ \ \ \forall\tilde{k}
\end{equation}
\item As $\rho\rightarrow\infty$, the areal coordinate is asymptotically
\begin{equation}
r(\rho)\simeq\rho-\frac{k^{2}r_{\text{s}}^{2}}{8\,\rho}
\end{equation}
\item As $\rho\rightarrow r_{\text{s}}$, for $\tilde{k}\neq0$, $\zeta$
is strictly greater than 1 and $r(\rho)\simeq\zeta r_{\text{s}}\left|1-\frac{r_{\text{s}}}{\rho}\right|^{\frac{1}{2}\left(\zeta-1\right)}\rightarrow0$.
All curves with $\tilde{k}\neq0$ have a zero at $\rho=r_{\text{s}}$
that separates the interior region, $\rho<r_{\text{s}}$, from the
exterior region, $\rho>r_{\text{s}}$.
\end{itemize}
The fact that the areal coordinate $r(\rho)$ shrinks to zero on the
 interior-exterior boundary if $\tilde{k}\neq0$ is a \emph{novel}
feature of pure $\mathcal{R}^{2}$ spacetime structures.

\noindent 
\begin{figure*}[!t]
\noindent \begin{centering}
\includegraphics{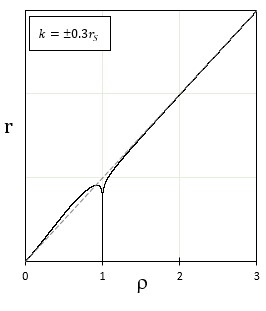}\includegraphics{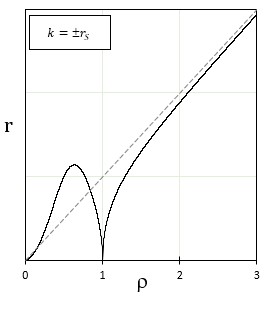}\includegraphics{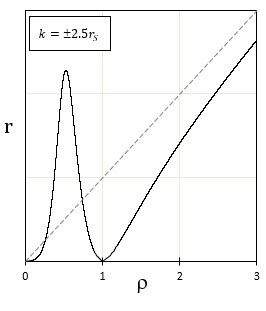}
\par\end{centering}
\caption{\label{fig:r-vs-rho}Areal coordinate $r$ as function of the new
coordinate $\rho$. For $\tilde{k}=\pm0.3,\pm1,\pm2.5$. In all plots,
$r_{\text{s}}=1$.}
\end{figure*}

\subsection{\label{subsec:Determinant}Determinant of the metric}

We next look into the determinant of metric \eqref{eq:special-B-1}--\eqref{eq:special-B-3},
\begin{align}
-\det g & =\left|1-\frac{r_{\text{s}}}{\rho}\right|^{4\tilde{k}}\frac{r^{8}(\rho)}{\rho^{4}}\,\sin^{2}\theta\label{eq:det-1}\\
 & =\frac{\zeta^{8}r_{\text{s}}^{8}}{\rho^{4}}\frac{\left|1-\frac{r_{\text{s}}}{\rho}\right|^{4(\zeta+\tilde{k}-1)}\,\sin^{2}\theta}{\left(1\mp\left|1-\frac{r_{\text{s}}}{\rho}\right|^{\zeta}\right)^{8}}\label{eq:det-2}
\end{align}
with $\mp$ corresponding the exterior/interior regions, respectively.
Fig. \ref{fig:aux-1} depicts a number of combinations of $\tilde{k}$
and $\zeta$ to be encountered in this paper. We deduce that 
\begin{equation}
\zeta+\tilde{k}-1=\begin{cases}
0 & \text{if }\tilde{k}=0\text{ or }\tilde{k}=-1\\
>0 & \text{if }\tilde{k}\in(-\infty,-1)\cup(0,+\infty)\\
<0 & \text{if }\tilde{k}\in(-1,0)
\end{cases}\label{eq:combo-1}
\end{equation}
\vskip4pt
\begin{figure}[!t]
\centering{}\includegraphics[scale=0.9]{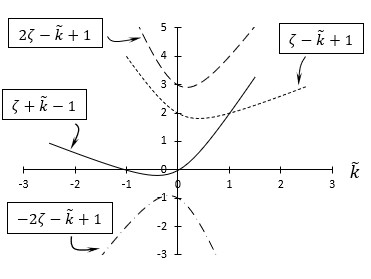}\caption{\label{fig:aux-1}Various combinations of $\tilde{k}$ and $\zeta$,
to be used in this paper, as functions of $\tilde{k}$.}
\end{figure}

\noindent \emph{Special cases:\vskip4pt}
\begin{itemize}
\item At $\tilde{k}=0$:
\begin{equation}
-\det g=\rho^{4}\sin^{2}\theta
\end{equation}
which is a result known in the Schwarzschild metric.
\item At $\tilde{k}=-1$:
\begin{equation}
-\det g=\frac{256\,r_{\text{s}}^{8}\,\sin^{2}\theta}{\rho^{4}\left(1\mp\left(1-\frac{r_{\text{s}}}{\rho}\right)^{2}\right)^{8}}
\end{equation}
with $\mp$ corresponding the exterior/interior regions, respectively.
The determinant with $\tilde{k}=-1$ is well-behaved for all $\rho\neq0$.
\end{itemize}
\noindent \emph{The asymptotic at the interior-exterior boundary,
$\rho\rightarrow r_{\text{s}}$:\vskip12pt}

Due to result \eqref{eq:combo-1}, we then have\small

\begin{equation}
\lim_{\rho\rightarrow r_{\text{s}}}\Bigl(-\det g|_{\theta=\frac{\pi}{2}}\Bigr)^{\frac{1}{4}}=\begin{cases}
\ r_{\text{s}} & \text{for }\tilde{k}=0\\
\ 4\,r_{\text{s}} & \text{for }\tilde{k}=-1\\
\ 0 & \text{for }\tilde{k}\leqslant-1\text{ or }\tilde{k}\geqslant0\\
\ +\infty & \text{for \ensuremath{\tilde{k}\in(-1,0)}}
\end{cases}
\end{equation}
\normalsize

\subsection{\label{subsec:Kretschmann}The Kretschmann invariant}

The Kretschmann scalar is given by\textcolor{black}{\small
\begin{align}
K & :=\mathcal{R}^{\mu\nu\rho\sigma}\mathcal{R}_{\mu\nu\rho\sigma}\label{eq:a-19a}\\
 & =\frac{2}{\zeta^{8}r_{\text{s}}^{4}}\left(1-\text{sgn}(x)\left|x\right|^{\zeta}\right)^{6}\left|x\right|^{2-4\zeta-2\tilde{k}}\times\nonumber \\
 & \biggl\{4\tilde{k}^{2}(\tilde{k}+1)\,\text{sgn}(x)\left|x\right|^{\zeta}+\zeta\,\Bigl(4\tilde{k}^{3}-5\tilde{k}^{2}-3\Bigr)\Bigl(1-\left|x\right|^{2\zeta}\Bigr)\nonumber \\
 & +\Bigl(9\tilde{k}^{4}-2\tilde{k}^{3}+10\tilde{k}^{2}+3\Bigr)\Bigl(1+\left|x\right|^{2\zeta}\Bigr)\biggr\}\label{eq:a-19b}
\end{align}
\normalsize in which $x:=1-\frac{r_{\text{s}}}{\rho}$.\vskip4pt}

By completing the square in the curly bracket in expression \eqref{eq:a-19b}
in terms of $\left|x\right|^{\zeta}$, one can show that the Kretschmann
scalar is positive-definite for all $\rho\in\mathbb{R}$ and all $k\in\mathbb{R}$.\vskip12pt

\noindent \emph{Special cases:\vskip4pt}
\begin{itemize}
\item At $\tilde{k}=0$, i.e. $\zeta=1$
\begin{equation}
K=\frac{12r_{\text{s}}^{2}}{\rho^{6}}
\end{equation}
recovering the result known in the Schwarzschild metric. It only has
a curvature singularity at the origin.
\item At $\tilde{k}=-1$, i.e. $\zeta=2$
\begin{equation}
K=\frac{3}{8r_{\text{s}}^{4}}\left(1\mp\left(1-\frac{r_{\text{s}}}{\rho}\right)^{2}\right)^{6}
\end{equation}
with $\mp$ corresponding to the exterior/interior regions, respectively.
It also only has a curvature singularity at the origin.
\end{itemize}
\noindent \emph{The asymptotics:\vskip12pt}
\begin{itemize}
\item As $\rho\rightarrow+\infty$, viz. $x\rightarrow1$,
\begin{align}
K & \simeq\frac{12}{\zeta^{8}r_{\text{s}}^{4}}\Bigl(\tilde{k}^{2}+1\Bigr)\Bigl(3\tilde{k}^{2}+1\Bigr)\left(1-\left|1-\frac{r_{\text{s}}}{\rho}\right|^{\zeta}\right)^{6}\\
 & \simeq12\Bigl(\tilde{k}^{2}+1\Bigr)\frac{r_{\text{s}}^{2}}{\rho^{6}}
\end{align}
which decays as $\rho^{-6}$ when $\rho\rightarrow+\infty$ for $\forall\tilde{k}\in\mathbb{R}$.
\item As $\rho\rightarrow0$, viz. $x\rightarrow\infty$,
\begin{align}
K & \simeq\frac{2}{\zeta^{8}r_{\text{s}}^{4}}\left|x\right|^{2\zeta-2\tilde{k}+2}\Biggl\{4\tilde{k}^{2}\left(\tilde{k}+1\right)\left|x\right|^{\zeta}+\nonumber \\
 & \left[\left(-4\tilde{k}^{3}+5\tilde{k}^{2}+3\right)\zeta+\left(9\tilde{k}^{4}-2\tilde{k}^{3}+10\tilde{k}^{2}+3\right)\right]\left|x\right|^{2\zeta}\Biggr\}
\end{align}
Since $\zeta=\left(1+3\tilde{k}^{2}\right)^{\frac{1}{2}}\geqslant1$
for $\forall\tilde{k}\in\mathbb{R}$, $\left|x\right|^{2\zeta}$ dominates
$\left|x\right|^{\zeta}$ as $x\rightarrow\infty$. Hence, as $\rho\rightarrow0$,
\small
\begin{align}
 & K\simeq\frac{2}{\zeta^{8}r_{\text{s}}^{4}}\left(\frac{r_{\text{s}}}{\rho}\right)^{2\left(2\zeta-\tilde{k}+1\right)}\times\nonumber \\
 & \left[\left(-4\tilde{k}^{3}+5\tilde{k}^{2}+3\right)\zeta+\left(9\tilde{k}^{4}-2\tilde{k}^{3}+10\tilde{k}^{2}+3\right)\right]
\end{align}
\normalsize From Fig. \ref{fig:aux-1}, $2\zeta-\tilde{k}+1>0\ \ \forall\tilde{k}\in\mathbb{R}$.
Thus $K$ diverges as $\rho^{-2\left(2\zeta-\tilde{k}+1\right)}$
when $\rho\rightarrow0$, for all $\forall\tilde{k}\in\mathbb{R}$.
\item As $\rho\rightarrow r_{\text{s}}$, viz. $x\rightarrow0$, if $\tilde{k}\neq0$
and $\tilde{k}\neq-1$
\begin{align}
 & K\simeq\frac{2}{\zeta^{8}r_{\text{s}}^{4}}\left|1-\frac{r_{\text{s}}}{\rho}\right|^{2\left(-2\zeta-\tilde{k}+1\right)}\times\nonumber \\
 & \left[\left(4\tilde{k}^{3}-5\tilde{k}^{2}-3\right)\zeta+\left(9\tilde{k}^{4}-2\tilde{k}^{3}+10\tilde{k}^{2}+3\right)\right]
\end{align}
From Fig. \ref{fig:aux-1}, $-2\zeta-\tilde{k}+1<0\ \ \forall k\in\mathbb{R}.$
Thus $K$ diverges as $\left|\rho-r_{\text{s}}\right|^{2\left(-2\zeta-\tilde{k}+1\right)}$
when $\rho\rightarrow r_{\text{s}}$, for $\tilde{k}\neq0$ and $\tilde{k}\neq-1$.
In sum, when $\rho\rightarrow r_{\text{s}}$,
\begin{equation}
K\simeq\begin{cases}
\ 12\,r_{\text{s}}^{-4} & \text{for }\tilde{k}=0\\
\ \frac{3}{8}\,r_{\text{s}}^{-4} & \text{for }\tilde{k}=-1\\
\ \left|\rho-r_{\text{s}}\right|^{2\left(-2\zeta-\tilde{k}+1\right)}\rightarrow+\infty & \text{otherwise}
\end{cases}\label{eq:K-asymptotic}
\end{equation}
\end{itemize}
The fact that, for $\tilde{k}\neq0$ and $\tilde{k}\neq-1$, the Kretschmann
scalar exhibits an additional singularity on the interior-exterior
boundary, $\rho=r_{\text{s}}$, besides the usual singularity at the
origin, is another \emph{novel} result.\vskip4pt

The plot for the Kretschmann scalar is shown in Fig. \ref{fig:Logarithm-of-Kretschmann}.
For clarity, we split the curves into two groups, one with negative
$\tilde{k}$ (upper panel), the other non-negative $\tilde{k}$ (lower
panel). The curves with $\tilde{k}=0$ and $\tilde{k}=-1$ are smooth
across the interior-exterior boundary, $\rho=r_{\text{s}}$. All other
curves show a divergence at $\rho=r_{\text{s}}$.
\begin{figure}[!t]
\noindent \begin{centering}
\includegraphics[scale=0.8]{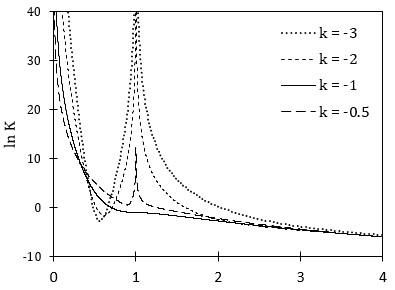}
\par\end{centering}
\noindent \begin{centering}
\includegraphics[scale=0.8]{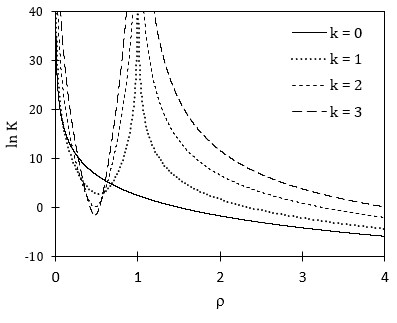}$\,$
\par\end{centering}
\caption{\label{fig:Logarithm-of-Kretschmann}Logarithm of Kretschmann invariant
as function of the new coordinate $\rho$, for various value of $\tilde{k}$.
For clarity, we plot the curves in two panels. In all cases, $r_{\text{s}}=1$.}
\end{figure}

\subsection{\label{subsec:Surface-area}Surface area of the interior-exterior
boundary of $\mathcal{R}^{2}$ spacetime: An anomalous behavior}

\begin{figure*}[!t]
\noindent \begin{centering}
\includegraphics[scale=0.9]{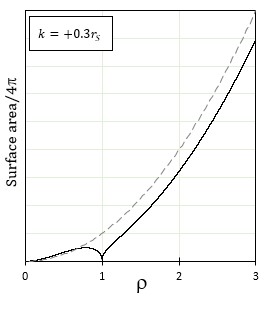}\includegraphics[scale=0.9]{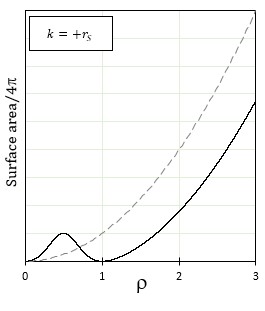}\includegraphics[scale=0.9]{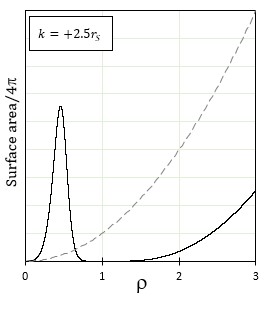}
\par\end{centering}
\noindent \begin{centering}
\includegraphics[scale=0.9]{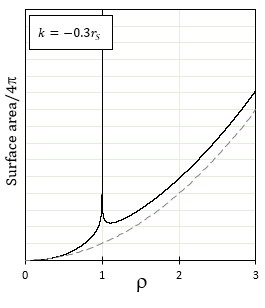}\includegraphics[scale=0.9]{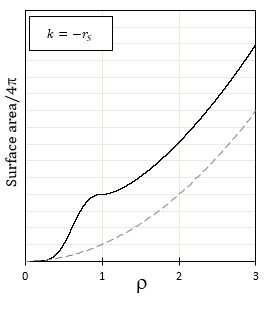}\includegraphics[scale=0.9]{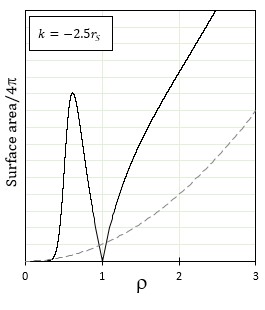}
\par\end{centering}
\caption{\label{fig:Surface-area}Surface area as function of the new coordinate
$\rho$. Upper panels: $\tilde{k}=0.3,1,2.5$. Lower panels: $\tilde{k}=-0.3,-1,-2.5$.
The dash line in each panel is the trivial $\tilde{k}=0$ benchmark,
$A=4\pi\rho^{2}$.}
\end{figure*}

For metric \eqref{eq:special-B-1}--\eqref{eq:special-B-2}, the
surface area of a two-dimensional sphere of ``radius'' $\rho$ is
\begin{align}
A & =4\pi\left|1-\frac{r_{\text{s}}}{\rho}\right|^{\tilde{k}}r^{2}(\rho)\nonumber \\
 & =4\pi\zeta^{2}r_{\text{s}}^{2}\frac{\left|1-\frac{r_{\text{s}}}{\rho}\right|^{\zeta+\tilde{k}-1}}{\left(1-\text{sgn}\left(1-\frac{r_{\text{s}}}{\rho}\right)\left|1-\frac{r_{\text{s}}}{\rho}\right|^{\zeta}\right)^{2}}
\end{align}
which is conveniently equal to
\begin{equation}
4\pi\rho\,\biggl(-\det g\Bigr|_{\theta=\pi/2}\biggr)^{\frac{1}{4}}
\end{equation}
The surface area $A$ and the determinant of $g$ thus share similar
behaviors. The plot of $A$ is shown in Fig. \ref{fig:Surface-area},
with the dotted parabola showing the regular $\tilde{k}=0$, in which
case $A=4\pi\rho^{2}$ since $r(\rho)=\rho$. Note the plots are not
symmetric with respect to $\tilde{k}$.\vskip12pt

\noindent \emph{The asymptotics:\vskip4pt}
\begin{itemize}
\item As $\rho\rightarrow+\infty$,
\begin{equation}
A\simeq4\pi\Bigl[\rho^{2}-\tilde{k}r_{\text{s}}\rho-\frac{r_{\text{s}}^{2}}{4}\tilde{k}(\tilde{k}-2)\Bigr]
\end{equation}
\item As $\rho\rightarrow0$,
\begin{equation}
A\simeq4\pi\zeta^{2}r_{\text{s}}^{2}\rho^{\zeta-\tilde{k}+1}
\end{equation}
From Fig. \ref{fig:aux-1}, $\zeta-\tilde{k}+1>0\ \ \forall\tilde{k}\in\mathbb{R}$.
Hence, $A\rightarrow0$ as $\rho\rightarrow0$ for $\forall\tilde{k}\in\mathbb{R}$.
\item As $\rho\rightarrow r_{\text{s}}$,
\begin{align}
A & \simeq4\pi\zeta^{2}r_{\text{s}}^{2}\left|1-\frac{r_{\text{s}}}{\rho}\right|^{\zeta+\tilde{k}-1}\\
 & =\begin{cases}
\ 4\pi r_{\text{s}}^{2} & \text{if }\tilde{k}=0\\
\ 16\pi r_{\text{s}}^{2} & \text{if }\tilde{k}=-1\\
\ 0 & \text{if }\tilde{k}\in(-\infty,-1)\cup(0,+\infty)\\
\ \text{+\ensuremath{\infty}} & \text{if }\tilde{k}\in\left(-1,0)\right)
\end{cases}\label{eq:surface-area}
\end{align}
\end{itemize}
\vskip6pt
\begin{rem}
Depending on the value of $\tilde{k}$, the shrinkage or divergence
of the surface area at $\rho=r_{\text{s}}$ is evident in Fig. \ref{fig:Surface-area}.
\end{rem}
\vskip6pt
\begin{rem}
The  interior-exterior boundary exhibits a peculiar property. Per
Eq. \eqref{eq:surface-area}, its surface area with $\tilde{k}\neq0$
drastically deviates from the customary $4\pi r_{\text{s}}^{2}$ expression,
thereby indicating that the Buchdahl parameter $k$ ``distorts''
the topology of spacetime around the interior-exterior boundary.
\end{rem}
\vskip6pt
\begin{rem}
The anomalous behavior of the surface area of the interior-exterior
boundary occurs in tandem with the curvature singularity at the interior-exterior
boundary in the Kretschmann invariant, Eq. \eqref{eq:K-asymptotic};
also see Section \ref{subsec:Non-Schwarzschild-structures}.
\end{rem}

\section{\label{sec:Application-II}Application II: $\ $Causal structure
of pure $\mathcal{R}^{2}$ spacetime}

This section analytically constructs the Kruskal-Szekeres (KS) diagram
of the \emph{special} Buchdahl-inspired metric attained in Lemma \ref{lem:lem-final}.
We adapt the usual practices that handle Schwarzschild black holes
-- by finding the tortoise coordinates, the Eddington-Finkelstein
coordinates, and the Kruskal-Szekeres coordinates \citep{Eddington-1924,Finkelstein-1958,Szekeres-1960,Kruskal-1960}
-- to the case at hand. Quantitative adjustments are needed. With
metric \eqref{eq:special-B-1}--\eqref{eq:special-B-3} involving
the parameter $\zeta$, we shall label these said coordinates by a
$\zeta-$ prefix. Fig. \ref{fig:KS-diagram} is the outcome of our
construction.

\subsection{\label{subsec:Tortoise}Constructing the $\boldsymbol{\zeta-}$tortoise
coordinate for pure $\mathcal{R}^{2}$ gravity}

The $\zeta-$tortoise coordinate $\rho^{*}(\rho)$ is defined as
\begin{equation}
d\rho^{*}:=\frac{r^{2}(\rho)}{\rho^{2}(1-\frac{r_{\text{s}}}{\rho})}d\rho
\end{equation}
\begin{align}
d\rho^{*} & =\zeta^{2}r_{\text{s}}^{2}\frac{\left|1-\frac{r_{\text{s}}}{\rho}\right|^{\zeta-1}\left(1-\frac{r_{\text{s}}}{\rho}\right)^{-1}}{\left(1-\text{sgn}(1-\frac{r_{\text{s}}}{\rho})\left|1-\frac{r_{\text{s}}}{\rho}\right|^{\zeta}\right)^{2}}\frac{d\rho}{\rho^{2}}\label{eq:tortoise-def}
\end{align}
The integral involves a Gaussian hypergeometric function. Let us define
\begin{equation}
z:=\text{sgn}\left(1-\frac{r_{\text{s}}}{\rho}\right)\left|1-\frac{r_{\text{s}}}{\rho}\right|^{\zeta}\label{eq:def-z}
\end{equation}

\paragraph{For $\rho>r_{\text{s}}$:}

\begin{align}
z & =\left(1-\frac{r_{\text{s}}}{\rho}\right)^{\zeta}>0\\
dz & =\zeta r_{\text{s}}\left(1-\frac{r_{\text{s}}}{\rho}\right)^{\zeta-1}\frac{d\rho}{\rho^{2}}
\end{align}
\begin{equation}
d\rho^{*}=\zeta r_{\text{s}}\frac{z^{-1/\zeta}}{(1-z)^{2}}dz
\end{equation}
giving (modulo an additive constant)
\begin{align}
\rho^{*} & =\frac{\zeta^{2}r_{\text{s}}}{\zeta-1}\,z^{1-\frac{1}{\zeta}}\,_{2}F_{1}\left(2,1-\frac{1}{\zeta};2-\frac{1}{\zeta};z\right)
\end{align}
\vskip4pt

\paragraph{For $0<\rho<r_{\text{s}}$:}

\begin{align}
z & =-\left(\frac{r_{\text{s}}}{\rho}-1\right)^{\zeta}<0\\
dz & =\zeta r_{\text{s}}\left(\frac{r_{\text{s}}}{\rho}-1\right)^{\zeta-1}\frac{d\rho}{\rho^{2}}
\end{align}
\begin{align}
d\rho^{*} & =-\zeta r_{\text{s}}\frac{\left(-z\right)^{-1/\zeta}}{(1-z)^{2}}dz\\
 & =\zeta r_{\text{s}}\frac{\left(-z\right)^{-1/\zeta}}{(1+(-z))^{2}}d(-z)\label{eq:d-tortois}
\end{align}
giving (modulo an additive constant)
\begin{align}
\rho^{*} & =\frac{\zeta^{2}r_{\text{s}}}{\zeta-1}\,\left(-z\right)^{1-\frac{1}{\zeta}}\,_{2}F_{1}\left(2,1-\frac{1}{\zeta};2-\frac{1}{\zeta};z\right)
\end{align}
\vskip4pt
\begin{figure}[!t]
\noindent \begin{centering}
\includegraphics[scale=0.8]{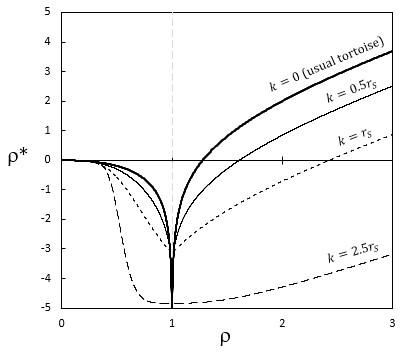}
\par\end{centering}
\caption{\label{fig:Tortoise-coordinate}The $\zeta-$tortoise coordinate of
Eq. \eqref{eq:tortoise-2} for various values of $\tilde{k}$ (with
$r_{\text{s}}=1$).}
\end{figure}

\noindent 
\begin{figure}[!t]
\noindent \begin{centering}
\includegraphics[scale=0.8]{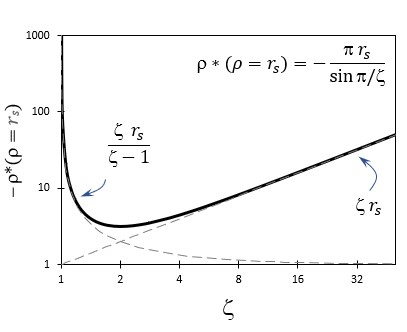}$\,$
\par\end{centering}
\caption{\label{fig:rho-star-0}$\rho^{*}(\rho=r_{\text{s}})$ as function
of $\zeta$; both axes in log scale (with $r_{\text{s}}=1$). The
two asymptopes are $\zeta/(\zeta-1)$ (for $\tilde{k}\rightarrow0$)
and $\zeta$ (for $\tilde{k}\rightarrow\infty$). Note that $\rho^{*}(\rho=r_{\text{s}})=-\pi$
for $\tilde{k}=1$.}
\end{figure}

In combination, we have the $\zeta-$tortoise coordinate (modulo an
additive constant) in terms of $z\in\mathbb{C}$ 
\begin{align}
\rho^{*} & =\frac{\zeta^{2}r_{\text{s}}}{\zeta-1}\,\left|z\right|^{1-\frac{1}{\zeta}}\,_{2}F_{1}\left(2,1-\frac{1}{\zeta};2-\frac{1}{\zeta};z\right)\label{eq:tortoise-1}
\end{align}
Furthermore, using Eq. \eqref{eq:d-tortois}, the difference
\begin{align}
\rho^{*}|_{\rho=0}-\rho^{*}|_{\rho=r_{\text{s}}} & =\int_{z=0}^{z=-\infty}\zeta r_{\text{s}}\frac{\left(-z\right)^{-1/\zeta}d(-z)}{(1+(-z))^{2}}\\
 & =\frac{\pi r_{\text{s}}}{\sin(\pi/\zeta)}\label{eq:diff-rho}
\end{align}
We shall choose the additive constant such that the $\zeta-$tortoise
coordinate vanishes at $\rho=0$. Using \eqref{eq:def-z} and \eqref{eq:diff-rho},
\eqref{eq:tortoise-1} produces\small

\begin{align}
\rho^{*} & =-\frac{\pi r_{\text{s}}}{\sin(\pi/\zeta)}+\frac{\zeta^{2}r_{\text{s}}}{\zeta-1}\left|1-\frac{r_{\text{s}}}{\rho}\right|^{\zeta-1}\times\nonumber \\
 & \ \ \ \,_{2}F_{1}\left(2,1-\frac{1}{\zeta};2-\frac{1}{\zeta};\text{sgn}\left(1-\frac{r_{\text{s}}}{\rho}\right)\left|1-\frac{r_{\text{s}}}{\rho}\right|^{\zeta}\right)\label{eq:tortoise-2}
\end{align}
\normalsize For $\tilde{k}\neq0$, the variable $\rho^{*}$ is continuous
across $\rho=r_{\text{s}}$ and $\rho^{*}|_{\rho=r_{\text{s}}}=-\frac{\pi r_{\text{s}}}{\sin(\pi/\zeta)}$.
In the complex plane $z\in\mathbb{C}$, the Gaussian hypergeometric
function $\,_{2}F_{1}\left(2,1-1/\zeta;2-1/\zeta;z\right)$ has a
branch point at $z=1$; expression \eqref{eq:tortoise-2} is thus
applicable for $z\in\mathbb{R}^{+}$ and $\tilde{k}\neq0$. See Appendix
\ref{sec:Gaussian-hypergeometric-function} for more information on
the hypergeometric function at play.\vskip4pt

For $\tilde{k}=0$, i.e. $\zeta=1$, the tortoise coordinate \eqref{eq:tortoise-2}
duly recovers 
\begin{equation}
\rho^{*}=\rho+r_{\text{s}}\ln\left|\frac{\rho-r_{\text{s}}}{r_{\text{s}}}\right|\label{eq:usual-tortoise}
\end{equation}
which diverges at $\rho=r_{\text{s}}$. See Appendix \ref{sec:Limit-k-0}
for derivation.\vskip4pt

Fig. \ref{fig:Tortoise-coordinate} plots the $\zeta-$tortoise coordinate
for various values of $\tilde{k}$. The case of $\tilde{k}=0$ is
the usual tortoise coordinate, Eq. \eqref{eq:usual-tortoise}. Fig.
\ref{fig:rho-star-0} shows the value $-\rho^{*}|_{\rho=r_{\text{s}}}=\frac{\pi r_{\text{s}}}{\sin(\pi/\zeta)}$
which asymptotes $\frac{\zeta r_{\text{s}}}{\zeta-1}$ for $\zeta\gtrsim1$
and $\zeta r_{\text{s}}$ for large $\zeta$.

\subsection{\label{subsec:EF-coordinates}Constructing the $\boldsymbol{\zeta-}$Eddington-Finkelstein
coordinates for pure $\mathcal{R}^{2}$ gravity}

Let us define the advanced and retarded $\zeta-$Eddington-Finkelstein
coordinates, per
\begin{align}
v & :=t+\rho^{*}\\
u & :=t-\rho^{*}
\end{align}
Metric \eqref{eq:special-B-1}, expressed in these new coordinates,
becomes \small
\begin{align}
ds^{2} & =\left|1-\frac{r_{\text{s}}}{\rho}\right|^{\tilde{k}}\times\nonumber \\
 & \biggl\{-\left(1-\frac{r_{\text{s}}}{\rho}\right)dv^{2}+\frac{r^{2}(\rho)}{\rho^{2}}\left(2dv\,d\rho+\rho^{2}d\Omega^{2}\right)\biggr\}
\end{align}
and
\begin{align}
ds^{2} & =\left|1-\frac{r_{\text{s}}}{\rho}\right|^{\tilde{k}}\times\nonumber \\
 & \biggl\{-\left(1-\frac{r_{\text{s}}}{\rho}\right)du^{2}+\frac{r^{2}(\rho)}{\rho^{2}}\left(-2du\,d\rho+\rho^{2}d\Omega^{2}\right)\biggr\}
\end{align}
\normalsize Also
\begin{equation}
du\,dv=dt^{2}-\frac{r^{4}(\rho)}{\rho^{4}\left(1-\frac{r_{\text{s}}}{\rho}\right)^{2}}d\rho^{2}
\end{equation}
thence
\begin{equation}
ds^{2}=\left|1-\frac{r_{\text{s}}}{\rho}\right|^{\tilde{k}}\biggl\{-\left(1-\frac{r_{\text{s}}}{\rho}\right)du\,dv+r^{2}(\rho)\,d\Omega^{2}\biggr\}
\end{equation}
In the advanced $\zeta-$Eddington-Finkelstein coordinate, the null
geodesics ($ds^{2}=0$) along the radial direction amount to
\begin{equation}
\frac{dv}{d\rho}=\begin{cases}
0 & \text{(infalling)}\\
\frac{2r^{2}(\rho)}{\rho^{2}\left(1-\frac{r_{\text{s}}}{\rho}\right)}=2\frac{d\rho^{*}}{d\rho} & \text{(outgoing)}
\end{cases}\label{eq:EF-slope}
\end{equation}
thus
\begin{equation}
v=\begin{cases}
\text{const} & \text{(infalling)}\\
2\rho^{*}+\text{const} & \text{(outgoing)}
\end{cases}
\end{equation}

\noindent 
\begin{figure*}[!t]
\noindent \begin{centering}
\includegraphics[scale=0.96]{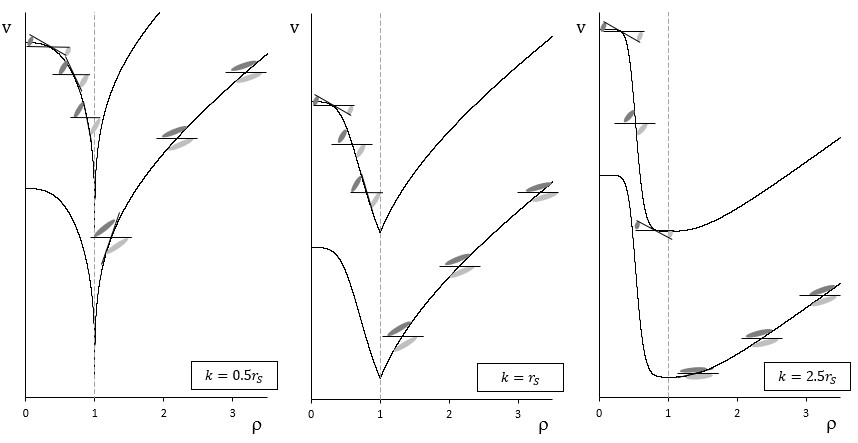}
\par\end{centering}
\caption{\label{fig:Lightcones}Light cones in the $(\rho,v)$ plane (with
$r_{\text{s}}=1$), for $\tilde{k}=0.5,1,2.5$. We choose these values
of $\tilde{k}$ as representatives for the three cases discussed in
the text.}
\end{figure*}

\subsection{\label{subsec:Behavior-of-lightcones}Behavior of light cones across
the interior-exterior boundary of a pure $\mathcal{R}^{2}$ spacetime}

In the advanced $\zeta-$Eddington-Finkelstein coordinates, per \eqref{eq:EF-slope}
and \eqref{eq:tortoise-def}, the outgoing null path has the slope\small
\begin{align}
\frac{dv}{d\rho} & =2\frac{d\rho^{*}}{d\rho}\\
 & =\frac{2\zeta^{2}}{1-\frac{r_{\text{s}}}{\rho}}\left(\frac{r_{\text{s}}}{\rho}\right)^{2}\frac{\left|1-\frac{r_{\text{s}}}{\rho}\right|^{\zeta-1}}{\left(1-\text{sgn}\left(1-\frac{r_{\text{s}}}{\rho}\right)\left|1-\frac{r_{\text{s}}}{\rho}\right|^{\zeta}\right)^{2}}
\end{align}
\normalsize Thus the outgoing null path exhibits the following asymptotic
behaviors
\begin{equation}
\frac{dv}{d\rho}\simeq\begin{cases}
\ -2\zeta^{2}r_{\text{s}}^{2(1-\zeta)}\rho^{\zeta}\rightarrow0 & \text{as }\rho\rightarrow0\\
\ +2\zeta^{2}r_{\text{s}}^{2-\zeta}\left|\rho-r_{\text{s}}\right|^{\zeta-2} & \text{as }\rho\rightarrow r_{\text{s}}^{+}\\
\ -2\zeta^{2}r_{\text{s}}^{2-\zeta}\left|\rho-r_{\text{s}}\right|^{\zeta-2} & \text{as }\rho\rightarrow r_{\text{s}}^{-}\\
\ +2 & \text{as }\rho\rightarrow\infty
\end{cases}
\end{equation}
Fig. \ref{fig:Lightcones} depicts the behavior of the light cones
in the $(v,\rho)$ plane. Concerning the light cone behavior across
the interior-exterior boundary, there are three cases:
\begin{casenv}
\item For $\left|\tilde{k}\right|<1$, viz. $\zeta<2$
\begin{equation}
\frac{dv}{d\rho}\rightarrow\pm\infty\text{ as }\rho\rightarrow r_{\text{s}}^{\pm}
\end{equation}
The light cone ``flips over'' across the interior-exterior boundary
as usual. This case includes the standard Schwarzschild metric, viz.
$\tilde{k}=0$. See the leftmost panel in Fig. \ref{fig:Lightcones}.
\item For $\left|\tilde{k}\right|>1$, viz. $\zeta>2$
\begin{equation}
\frac{dv}{d\rho}\rightarrow0^{\pm}\text{ as }\rho\rightarrow r_{\text{s}}^{\pm}
\end{equation}
This case is a peculiar situation. The light cone first ``flattens
out'' when approaching the interior-exterior boundary from the exterior.
Upon passing the interior-exterior boundary, the light cone makes
sudden ``collapse'' to an single line, $dv=0$, then gradually ``re-widens''
when entering into the interior. See the rightmost panel in Fig. \ref{fig:Lightcones}.
\item For $\left|\tilde{k}\right|=1$, hence $\zeta=2$
\begin{equation}
\frac{dv}{d\rho}\rightarrow\pm8\text{ as }\rho\rightarrow r_{\text{s}}^{\pm}
\end{equation}
The light cones changes its slope in a step-wise fashion. See the
middle panel in Fig. \ref{fig:Lightcones}.
\end{casenv}
\begin{rem}
\label{rem:lightcones}From Fig. \ref{fig:Lightcones}, in every situation,
all light paths and time-like paths inside the interior-exterior boundary
would eventually reach the origin; they cannot escape from the interior.
In the exterior region, the outgoing light path can escape to infinity.
These results shall be confirmed by way of the Kruskal-Szekeres diagram
in Sec. \ref{subsec:KS-diagram} below.
\end{rem}

\subsection{\label{subsec:KS-diagram}Constructing the $\boldsymbol{\zeta-}$Kruskal-Szekeres
coordinates for pure $\mathcal{R}^{2}$ gravity}

Most of the procedure originally advanced by Kruskal and Szekeres
for Schwarzschild black holes \citep{Kruskal-1960,Szekeres-1960}
can be re-purposed for pure $\mathcal{R}^{2}$ spacetime. We shall
consider the exterior and interior regions separately.

\subsubsection*{\textbf{The exterior}}

For $\rho>r_{\text{s}}$, let us define
\begin{align}
X & :=\frac{1}{2}\left(e^{\frac{v}{2r_{\text{s}}}}+e^{-\frac{u}{2r_{\text{s}}}}\right)\\
T & :=\frac{1}{2}\left(e^{\frac{v}{2r_{\text{s}}}}-e^{-\frac{u}{2r_{\text{s}}}}\right)
\end{align}
then
\begin{align}
X & =e^{\frac{\rho^{*}}{2r_{\text{s}}}}\cosh\frac{t}{2r_{\text{s}}}\\
T & =e^{\frac{\rho^{*}}{2r_{\text{s}}}}\sinh\frac{t}{2r_{\text{s}}}
\end{align}
\begin{align}
T^{2}-X^{2} & =-e^{\frac{\rho^{*}}{r_{\text{s}}}}\\
\frac{T}{X} & =\tanh\frac{t}{2r_{\text{s}}}
\end{align}
and \small
\begin{align}
dX & =\frac{e^{\frac{\rho^{*}}{2r_{\text{s}}}}}{2r_{\text{s}}}\left[\frac{r^{2}(\rho)}{\rho^{2}\left(1-\frac{r_{\text{s}}}{\rho}\right)}\cosh\frac{t}{2r_{\text{s}}}d\rho+\sinh\frac{t}{2r_{\text{s}}}dt\right]\\
dT & =\frac{e^{\frac{\rho^{*}}{2r_{\text{s}}}}}{2r_{\text{s}}}\left[\frac{r^{2}(\rho)}{\rho^{2}\left(1-\frac{r_{\text{s}}}{\rho}\right)}\sinh\frac{t}{2r_{\text{s}}}d\rho+\cosh\frac{t}{2r_{\text{s}}}dt\right]
\end{align}
hence
\begin{equation}
dT^{2}-dX^{2}=\frac{e^{\frac{\rho^{*}}{r_{\text{s}}}}}{4r_{\text{s}}^{2}}\left[dt^{2}-\frac{r^{4}(\rho)}{\rho^{4}\left(1-\frac{r_{\text{s}}}{\rho}\right)^{2}}d\rho^{2}\right]
\end{equation}
giving
\begin{align}
ds^{2} & =\left|1-\frac{r_{\text{s}}}{\rho}\right|^{\tilde{k}}\times\nonumber \\
 & \left\{ -4r_{\text{s}}^{2}e^{-\frac{\rho^{*}}{r_{\text{s}}}}\left(1-\frac{r_{\text{s}}}{\rho}\right)\left(dT^{2}-dX^{2}\right)+r^{2}(\rho)d\Omega^{2}\right\} 
\end{align}
\normalsize

\subsubsection*{\textbf{The interior}}

For $\rho<r_{\text{s}}$, let us define
\begin{align}
X & :=\frac{1}{2}\left(e^{\frac{v}{2r_{\text{s}}}}-e^{-\frac{u}{2r_{\text{s}}}}\right)\\
T & :=\frac{1}{2}\left(e^{\frac{v}{2r_{\text{s}}}}+e^{-\frac{u}{2r_{\text{s}}}}\right)
\end{align}
then
\begin{align}
X & =e^{\frac{\rho^{*}}{2r_{\text{s}}}}\sinh\frac{t}{2r_{\text{s}}}\\
T & =e^{\frac{\rho^{*}}{2r_{\text{s}}}}\cosh\frac{t}{2r_{\text{s}}}
\end{align}
\begin{align}
T^{2}-X^{2} & =+e^{\frac{\rho^{*}}{r_{\text{s}}}}\\
\frac{T}{X} & =\left(\tanh\frac{t}{2r_{\text{s}}}\right)^{-1}
\end{align}
and \small
\begin{align}
dX & =\frac{e^{\frac{\rho^{*}}{2r_{\text{s}}}}}{2r_{\text{s}}}\left[\frac{r^{2}(\rho)}{\rho^{2}\left(1-\frac{r_{\text{s}}}{\rho}\right)}\sinh\frac{t}{2r_{\text{s}}}d\rho+\cosh\frac{t}{2r_{\text{s}}}dt\right]\\
dT & =\frac{e^{\frac{\rho^{*}}{2r_{\text{s}}}}}{2r_{\text{s}}}\left[\frac{r^{2}(\rho)}{\rho^{2}\left(1-\frac{r_{\text{s}}}{\rho}\right)}\cosh\frac{t}{2r_{\text{s}}}d\rho+\sinh\frac{t}{2r_{\text{s}}}dt\right]
\end{align}
hence
\begin{equation}
dT^{2}-dX^{2}=-\frac{e^{\frac{\rho^{*}}{r_{\text{s}}}}}{4r_{\text{s}}^{2}}\left[dt^{2}-\frac{r^{4}(\rho)}{\rho^{4}\left(1-\frac{r_{\text{s}}}{\rho}\right)^{2}}d\rho^{2}\right]
\end{equation}
giving
\begin{align}
ds^{2} & =\left|1-\frac{r_{\text{s}}}{\rho}\right|^{\tilde{k}}\times\nonumber \\
 & \left\{ +4r_{\text{s}}^{2}e^{-\frac{\rho^{*}}{r_{\text{s}}}}\left(1-\frac{r_{\text{s}}}{\rho}\right)\left(dT^{2}-dX^{2}\right)+r^{2}(\rho)d\Omega^{2}\right\} 
\end{align}
\normalsize

\subsubsection*{\textbf{Combination of both regions}}

The \emph{special} Buchdahl-inspired metric in the $\zeta-$Kruskal-Szekeres
(KS) coordinates is thus \vskip-8pt

\small

\begin{align}
ds^{2} & =\left|1-\frac{r_{\text{s}}}{\rho}\right|^{\tilde{k}}\times\nonumber \\
 & \left\{ -4r_{\text{s}}^{2}e^{-\frac{\rho^{*}}{r_{\text{s}}}}\left|1-\frac{r_{\text{s}}}{\rho}\right|\left(dT^{2}-dX^{2}\right)+r^{2}(\rho)d\Omega^{2}\right\} \label{eq:KS-metric}
\end{align}
\normalsize and
\begin{align}
T^{2}-X^{2} & =-\text{sgn(\ensuremath{\rho-r_{\text{s}})}}e^{\frac{\rho^{*}}{r_{\text{s}}}}\label{eq:T-X-1}\\
\frac{T}{X} & =\left(\tanh\frac{t}{2r_{\text{s}}}\right)^{\text{sgn(\ensuremath{\rho-r_{\text{s}})}}}\label{eq:T-X-2}
\end{align}
\vskip12pt
\begin{rem}
For the case $\tilde{k}=0$, substituting $\rho^{*}=\rho+r_{\text{s}}\ln\left|\frac{\rho}{r_{\text{s}}}-1\right|$
and $r(\rho)=\rho$ into \eqref{eq:KS-metric}, we get\small
\begin{equation}
ds_{\text{ (KS)}}^{2}=-4r_{\text{s}}^{3}\frac{e^{-\frac{\rho}{r_{\text{s}}}}}{\rho}\left(dT^{2}-dX^{2}\right)+\rho^{2}d\Omega^{2}\label{eq:Schwa}
\end{equation}
which is the usual KS result for Schwarzschild black holes. \normalsize
\end{rem}

\subsection{\label{subsec:Features}Features of the $\boldsymbol{\zeta-}$Kruskal-Szekeres
diagram: A new \textquotedblleft virtual\textquotedblright{} region}

\begin{figure*}[!t]
\noindent \begin{centering}
\includegraphics[scale=0.7]{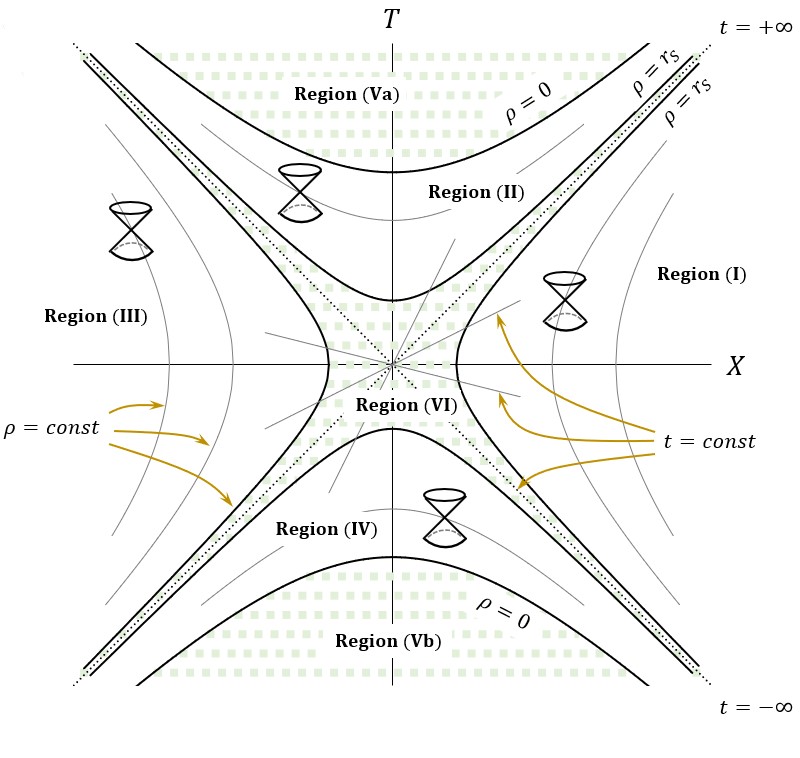}
\par\end{centering}
\caption{\label{fig:KS-diagram}Kruskal-Szekeres diagram for $\tilde{k}\protect\neq0$.
The \textquotedblleft gulf\textquotedblright{} shown as Region (VI)
is a new feature. See text for explanations.}
\end{figure*}

Restricting to the radial direction, viz. $d\theta=d\phi=0$, metric
\eqref{eq:KS-metric} is
\begin{equation}
ds^{2}=-4r_{\text{s}}^{2}e^{-\frac{\rho^{*}}{r_{\text{s}}}}\left|1-\frac{r_{\text{s}}}{\rho}\right|^{1+\tilde{k}}\left(dT^{2}-dX^{2}\right)\label{eq:zeta-KS}
\end{equation}
The $\zeta-$Kruskal-Szekeres ($\zeta-$KS for short) plane for metric
\eqref{eq:zeta-KS} is shown in Fig. \ref{fig:KS-diagram}. A number
of key features are:
\begin{itemize}
\item Similarly to the usual KS diagram, the $\zeta-$KS diagram is conformally
Minkowski.
\item The null geodesics are:
\begin{equation}
dX=\pm dT
\end{equation}
Light thus travels on the $45^{\circ}$ lines in the $\zeta-$KS plane. 
\item The $\zeta-$KS diagram retains, qualitatively, most features of the
causal structure established for the usual KS diagram. There are quantitative
changes; see below.
\item A constant--$\rho$ contour corresponds to a hyperbola, whereas the
constant--$t$ contour to a straight line through the origin of the
$(T,X)$ plane.
\item The coordinate origin $\rho=0$ amounts to, per 
\begin{equation}
T^{2}-X^{2}=1
\end{equation}
because $\rho^{*}(\rho=0)=0$.
\item The interior-exterior boundary $\rho=r_{\text{s}}$ amounts to \emph{two}
\emph{distinct} hyperbolae, one for the interior and the other the
exterior, per
\begin{equation}
T^{2}-X^{2}=\begin{cases}
-e^{-\frac{\pi}{\sin\frac{\pi}{\zeta}}} & \text{for exterior}\\
+e^{-\frac{\pi}{\sin\frac{\pi}{\zeta}}} & \text{for interior}
\end{cases}\label{eq:hyperbolae}
\end{equation}
Note that each hyperbola comprises of two separate branches on its
own.
\item For $\tilde{k}=0$, viz. $\zeta=1$, the hyperbolae \eqref{eq:hyperbolae}
degenerate to two straight lines
\begin{equation}
T=\pm X
\end{equation}
as expected for Schwarzschild black holes.
\item Region (I) is the exterior, mapped into the $\zeta-$KS plane extended
up to the right branch of the hyperbola $T^{2}-X^{2}=-e^{-\frac{\pi}{\sin\frac{\pi}{\zeta}}}$.
\item Region (II) is the interior, mapped into the $\zeta-$KS plane, extended
up to the upper branch of the hyperbola $T^{2}-X^{2}=+e^{-\frac{\pi}{\sin\frac{\pi}{\zeta}}}$.
\item Regions (III) and (IV) are time-reverse images of Regions (I) and
(II).
\item Regions (Va) and (Vb) (shaded by dots) are unphysical, viz. $\rho<0$.
\item What is new is Region (VI) (also shaded in dots) that sandwiches between
the four hyperbola branches given by \eqref{eq:hyperbolae}.
\end{itemize}
\noindent In Region (II), all timelike and null trajectories will
eventually hit the origin, denoted by the hyperbola, $\rho=0$. Nothing
can escape from the interior. In Region (I), outgoing light paths
would be able to escape to infinity. These observations are in agreement
with the result obtained in Sec. \ref{subsec:Behavior-of-lightcones};
see Remark \ref{rem:lightcones}.\vskip4pt

Incoming light paths from Region (I) must enter Region (II) by ``\emph{bypassing}''
Region (VI). An infalling object (or light wave) would hit the interior-exterior
boundary $\rho=r_{\text{s}}$ on the side of Region (I) then reappear
on the interior-exterior boundary on the side of Region (II). The
``transit'' -- if there is any -- within Region (VI) is not visible,
thus ``virtual'', for an outside observer from afar.\vskip4pt

Region (VI) appears as a ``gulf'' in the $(T,X)$ coordinate system
but it does not correspond to any region in the $(t,\rho)$ coordinate
system. When $\tilde{k}\rightarrow0$, the ``gulf'' shrinks toward
the 2 lines $T=\pm X$. Given that the $\zeta-$KS diagram is the
maximal extension of the \emph{special} Buchdahl-inspired metric,
the emergence of Region (VI) is a highly curious feature, signaling
potential new physics that takes place on the interior-exterior boundary
of $\mathcal{R}^{2}$ spacetime. Taken altogether, the singularity
on the interior-exterior boundary in the Kretschmann invariant, the
anomalous behavior of the surface area of the interior-exterior boundary,
and the ``gulf'' in the $\zeta-$KS diagram indicate that the topology
of $\mathcal{R}^{2}$-gravity spacetime around a mass source undergoes
fundamental alterations when the Buchdahl parameter $k$ is in presence. 

\subsection{\label{subsec:A-conjecture}A conjecture}

While the intuitions about the causal structure built for the usual
KS diagram remain intact for its $\zeta-$KS enlargement, the appearance
of the ``virtual'' Region (VI) would beg for further examinations.
We shall venture some ideas going forward.\vskip4pt

Let us recall that in the usual KS diagram, the tortoise coordinate
is ``bifurcated'' into two branches, separately for the exterior
and for the interior, per
\begin{equation}
\rho^{*}=\begin{cases}
\rho+r_{\text{s}}\ln\left(\rho-r_{\text{s}}\right) & \text{for exterior}\\
\rho+r_{\text{s}}\ln\left(r_{\text{s}}-\rho\right) & \text{for interior}
\end{cases}
\end{equation}
For the $\zeta-$tortoise coordinate obtained in Sec. \ref{subsec:Tortoise},
this ``bifurcation'' issue is somewhat mitigated if $\tilde{k}\neq0$,
viz. $\zeta>1$. To see this, let us recall Eqs. \eqref{eq:def-z}
and \eqref{eq:tortoise-1} with the additive constant term being suppressed
for convenience

\begin{align}
\rho^{*} & =\frac{\zeta^{2}r_{\text{s}}}{\zeta-1}\left|z\right|^{1-\frac{1}{\zeta}}\,_{2}F_{1}\left(2,1-\frac{1}{\zeta};2-\frac{1}{\zeta};z\right)\label{eq:tortoise-3}\\
z & :=\text{\text{sgn}\ensuremath{\left(1-\frac{r_{\text{s}}}{\rho}\right)\left|1-\frac{r_{\text{s}}}{\rho}\right|^{\zeta}}}\label{eq:tortoise-4}
\end{align}
in which $z\in\mathbb{R}$ (here, we consider $\rho\in\mathbb{R}$
unrestricted). The Gaussian hypergeometric function $\,_{2}F_{1}\left(2,1-1/\zeta;2-1/\zeta;z\right)$,
when extended onto the complex plane $z\in\mathbb{C}$, has a branch
point at $z=1$ (corresponding to $\rho=\pm\infty)$. For $\tilde{k}=0$,
Eqs. \eqref{eq:tortoise-3}--\eqref{eq:tortoise-4} recovers the
usual tortoise coordinate (see Appendix \ref{sec:Limit-k-0}):
\begin{align}
\rho^{*} & =\frac{r_{\text{s}}}{1-z}+r_{\text{s}}\ln\left|\frac{z}{1-z}\right|\\
 & =\rho+r_{\text{s}}\ln\left|\frac{\rho}{r_{\text{s}}}-1\right|
\end{align}
which is \emph{not} analytic across $z=0$, a point that separates
the exterior from the interior, as alluded to above.\vskip4pt

To proceed, let us define the following auxiliary variable for $z\in\mathbb{C}$,
\begin{align}
\tilde{\rho} & :=\frac{\zeta^{2}r_{\text{s}}}{\zeta-1}\,z^{1-\frac{1}{\zeta}}\,_{2}F_{1}\left(2,1-\frac{1}{\zeta};2-\frac{1}{\zeta};z\right)
\end{align}
in which the $z^{1-\frac{1}{\zeta}}$ term has replaced the $\left|z\right|^{1-\frac{1}{\zeta}}$
term in Eq. \eqref{eq:tortoise-3}. The $\zeta-$tortoise coordinate
is thus

\begin{equation}
\rho^{*}(z)=\left(\frac{\left|z\right|}{z}\right)^{1-\frac{1}{\zeta}}\tilde{\rho}(z)=e^{-i\left(1-\frac{1}{\zeta}\right)\arg z}\tilde{\rho}(z)\label{eq:tortoise-5}
\end{equation}
which, when restricted to $z\in\mathbb{R}$, yields two separate branches
\begin{equation}
\rho^{*}(z)=\begin{cases}
\tilde{\rho}(z) & \text{exterior}\\
e^{-i\left(1-\frac{1}{\zeta}\right)\pi}\tilde{\rho}(z) & \text{interior}
\end{cases}
\end{equation}
The variable $\tilde{\rho}$, when defined in the complex plane $z\in\mathbb{C}$,
might be used to ``analytically continue'' from the interior ($z\in\mathbb{R}^{-}$)
to the exterior ($z\in\mathbb{R}^{+}$). In the meantime, the phase
factor $e^{-i\left(1-\frac{1}{\zeta}\right)\arg z}$ in Eq. \eqref{eq:tortoise-5}
isolates the non-analytical part in $\rho^{*}$ from the ``well-behaved''
$\tilde{\rho}$, hence lessening the ``bifurcation'' issue mentioned
above.\vskip4pt

Concerning $\tilde{\rho}$, for a general value of $\tilde{k}\neq0$,
the exponent $1-\frac{1}{\zeta}$ is strictly confined within the
range $(0,1)$; the term $z^{1-\frac{1}{\zeta}}$ is thus multi-valued
and the $z=0$ point represents a branch point. (N.B: the function
$\,_{2}F_{1}$ itself contains another branch point at $z=1$.)\vskip4pt

We conjecture that the variable $\tilde{\rho}$, defined as a function
of $z$ in the complex plane $\mathbb{C}$, could serve as a tool
to tackle the ``gulf'' in the $\zeta-$KS diagram, a topic worthwhile
of future research.
\begin{conjecture}
\label{conj:my-conjecture}The auxiliary variable
\begin{align}
\tilde{\rho} & :=\frac{\zeta^{2}r_{\text{s}}}{\zeta-1}\,z^{1-\frac{1}{\zeta}}\,_{2}F_{1}\left(2,1-\frac{1}{\zeta};2-\frac{1}{\zeta};z\right)
\end{align}
with $z\in\mathbb{C}\smallsetminus\mathbb{R}$, viz. $\text{Im }z\neq0$,
represents the ``virtual'' Region (VI) in the $\zeta-$Kruskal-Szekeres
diagram.
\end{conjecture}

\section{\label{sec:Summary}Summary and outlooks}

\textcolor{black}{Lemma \ref{lem:lem-final} in Sec. \ref{subsec:The-special-Buchdahl-inspired}
is the central result of our work,\linebreak finalizing the program
that Hans A. Buchdahl pioneered -- but prematurely abandoned --
circa 1962 \citep{Buchdahl-1962}. It presents an asymptotically flat
non-Schwarzschild spacetime in exact closed analytical form, which
we reproduce here for the reader's convenience \vskip-10pt}

\small
\begin{equation}
\left|1-\frac{r_{\text{s}}}{\rho}\right|^{\frac{k}{r_{\text{s}}}}\biggl\{-\Bigl(1-\frac{r_{\text{s}}}{\rho}\Bigr)dt^{2}+\frac{r^{4}(\rho)\,d\rho^{2}}{\rho^{4}\Bigl(1-\frac{r_{\text{s}}}{\rho}\Bigr)}+r^{2}(\rho)d\Omega^{2}\biggr\}\label{eq:metric-summary-1}
\end{equation}
\normalsize The areal coordinate $r$ is related to the radial coordinate
$\rho$ per \small
\begin{align}
r(\rho) & :=\frac{\zeta\,r_{\text{s}}\left|1-\frac{r_{\text{s}}}{\rho}\right|^{\frac{1}{2}(\zeta-1)}}{\left|1-\text{sgn}\left(1-\frac{r_{\text{s}}}{\rho}\right)\left|1-\frac{r_{\text{s}}}{\rho}\right|^{\zeta}\right|}\label{eq:metric-summary-2}
\end{align}
\normalsize with $\zeta:=\sqrt{1+3k^{2}/r_{\text{s}}^{2}}$ (we have
restored $k:=\tilde{k}\,r_{\text{s}}$).\vskip6pt

The \emph{special} Buchdahl-inspired metric is a member of the branch
of non-trivial solutions, viz. the class of Buchdahl-inspired metrics
with $\Lambda\in\mathbb{R}$, obtained in our preceding work for pure
$\mathcal{R}^{2}$ gravity \citep{Nguyen-2022-Buchdahl}; also see
Eqs. \eqref{eq:B-metric-1}--\eqref{eq:B-metric-4} in this current
paper. Fig. \vref{fig:Buchdahl-inspired-metric-family} summarizes
the state of affairs: the generic Buchdahl-inspired metric with $\Lambda\in\mathbb{R}$
supersedes the Schwarzschild-de Sitter metric and the \emph{special}
Buchdahl-inspired metric supersedes the Schwarzschild metric. Both
of the superseding instants occur when the Buchdahl parameter $k$
is sent to zero. \footnote{In comparison, the L\"u-Perkins-Pope-Stelle solution in Einstein-Weyl
gravity is a second branch \emph{separate} from the Schwarzschild
branch \citep{Lu-2015-a,Lu-2015-b}.}

\subsubsection{Higher-derivative characteristic}

The asymptotically flat $\mathcal{R}^{2}$ spacetime, described by
metric \eqref{eq:metric-summary-1}--\eqref{eq:metric-summary-2},
is characterized by a ``Schwarzschild'' radius $r_{\text{s}}$ and
the Buchdahl parameter $k$, the latter of which stems from the higher-order
nature of the quadratic theory. If $\mathcal{R}^{2}$ spacetime structures
shall eventually have been proven to be stable \citep{Goldstein-2017,Rinaldi-2020,Held-2021},
then the Buchdahl parameter $k$ would represent new higher-derivative
characteristic in addition to the mass of the source (encoded by $r_{\text{s}}$)
\footnote{The angular momentum and electric charge of the source are not active
in our consideration here.}.\vskip4pt

Furthermore, being a signature of higher-order theory, the Buchdahl
parameter $k$ should leave its footprints in higher-derivative gravity
at large. In the companion paper \citep{Nguyen-2022-extension}, we
confirm this intuition by carrying the concept of a Buchdahl parameter
over to the quadratic action $\mathcal{R}^{2}+\gamma\left(\mathcal{R}-2\Lambda\right)$;
therein we found a new vacuo which depends on $k$ as a perturbative
parameter. The Buchdahl parameter therefore should be a generic universal
hallmark of several modified theories of gravity.

\subsubsection{Relevance of the metric}

\textcolor{black}{A metric that is merely Ricci-scalar-flat is an
automatic }\textcolor{black}{\emph{trivial}}\textcolor{black}{{} solution
to the pure $\mathcal{R}^{2}$ vacuo field equation. Such as metric
is under-determined, though, as it is subject to only }\textcolor{black}{\emph{one}}\textcolor{black}{{}
single constraint, viz. $\mathcal{R}=0$, which is not sufficient
to determine the full $g_{\mu\nu}$ metric. Examples of null-Ricci-scalar
metrics hence are in abundance; some are given, e.g., in \citep{Shankaranarayanan-2018}.\vskip4pt}

\textcolor{black}{Yet, despite its null Ricci scalar}, the \emph{special}
Buchdahl-inspired metric \eqref{eq:metric-summary-1}--\eqref{eq:metric-summary-2}
acquires its structure by being a member of the class of \emph{non-trivial}
solutions, the Buchdahl-inspired metrics given in Eqs. \eqref{eq:B-metric-1}--\eqref{eq:B-metric-4}.
The Venn diagrams in Fig. \vref{fig:Buchdahl-inspired-metric-family}
depict the relations among the various metrics in question.\vskip4pt

\textcolor{black}{The }\textcolor{black}{\emph{special}}\textcolor{black}{{}
Buchdahl-inspired metric describes asymptotically flat spacetimes,
a situation with theoretical appeal in and of itself. Yet it remains
of relevance for asymptotically constant spacetimes in general. For
a generic $\Lambda\neq0$, in the range of $r\ll\left|\Lambda\right|^{-\frac{1}{2}}$,
the $\Lambda\,r^{2}$ term in the evolution rule \eqref{eq:B-metric-3}
would be suppressed. This means that the }\textcolor{black}{\emph{special}}\textcolor{black}{{}
Buchdahl-inspired metric still works well }\textcolor{black}{\emph{deep
inside the bulk}}\textcolor{black}{{} for a generic Buchdahl-inspired
spacetime with $\Lambda\neq0$. That is to say, in all practical situations,
pure $\mathcal{R}^{2}$ structures (whether they live on an asymptotically
flat or an asymptotically constant background) are well described
by }metric \eqref{eq:metric-summary-1}--\eqref{eq:metric-summary-2},
and the anomalous properties of $\mathcal{R}^{2}$ spacetime, discovered
herein and summarized below, remain valid as long as $\left|\Lambda\,r_{\text{s}}^{2}\right|\ll1$\textcolor{black}{.\vskip4pt}

\textcolor{black}{Asymptotically flat non-Schwarzschild solutions
that are non-trivial (in the sense of not being under-determined)
in modified gravity are a rare bread. An intriguing example is the
L\"u-Perkins-Pope-Stelle solution in Einstein-Weyl gravity \citep{Lu-2015-a,Lu-2015-b}.
In \citep{Kalita-2019} Kalita and Mukhopadhyay also reported numerical
indications of an asymptotically flat vacuo for an $f(\mathcal{R})$
theory with the Einstein-Hilbert $\mathcal{R}$ being the leading
term. The }\textcolor{black}{\emph{special}}\textcolor{black}{{} Buchdahl-inspired
metric, found in our current paper, is a newest member of this scant
roster.}

\subsubsection{Anomalous behavior in the surface area of the interior-exterior boundary}

Equipped with the exact analytical solution \eqref{eq:metric-summary-1}--\eqref{eq:metric-summary-2},
we then examined asymptotically flat $\mathcal{R}^{2}$ spacetime
structures. We found that, except for $k=0$, the areal radius $r(\rho)$
shrinks to zero at the interior-exterior boundary. See Sec. \ref{subsec:Areal-coordinate}.\vskip4pt

Crucially, we also found that the surface area of the interior-exterior
boundary, by including the conformal factor $\left|1-\frac{r_{\text{s}}}{\rho}\right|^{\frac{k}{r_{\text{s}}}}$,
vanishes for $k\in(-\infty,-r_{\text{s}})\cup(0,+\infty)$, diverges
for $k\in(-r_{\text{s}},0)$, equal $4\pi r_{\text{s}}^{2}$ for $k=0$,
and equal $16\pi r_{\text{s}}^{2}$ for $k=-r_{\text{s}}$. See Sec.
\ref{subsec:Surface-area}.\vskip4pt

\textcolor{black}{At the same time, the Kretschmann invariant exhibits
curvature singularities on the interior-exterior boundary provided
that $k\neq0$ and $k\neq-r_{\text{s}}$. The usual singularity the
origin persists, but it gets modified in the presence of $k$. See
Sec. \ref{subsec:Kretschmann}.\vskip4pt}

Taken altogether, these anomalous properties of the interior-exterior
boundary suggest that the topology of $\mathcal{R}^{2}$ spacetimes
undergo fundamental changes around mass sources.

\subsubsection{A \textquotedblleft virtual\textquotedblright{} region in the $\zeta-$Kruskal-Szekeres
diagram}

We proceeded by analytically construct the KS diagram for metric \eqref{eq:metric-summary-1}--\eqref{eq:metric-summary-2}.
The techniques developed for the regular KS diagram \citep{Szekeres-1960,Kruskal-1960,Eddington-1924,Finkelstein-1958}
are extendable to the\linebreak case at hand. We employed them to
design the $\zeta-$tortoise coordinate, the $\zeta-$Eddington-Finkelstein
coordinates, and the $\zeta-$Kruskal-Szekeres coordinates, accordingly.
\vskip4pt

The $\zeta-$tortoise coordinate $\rho^{*}$ is expressible in terms
of a Gaussian hypergeometric function. We found modifications in the
``flip over'' phenomenon of light cones across the interior-exterior
boundary. \textcolor{black}{See Secs. \ref{subsec:Tortoise} and \ref{subsec:Behavior-of-lightcones}.\vskip4pt}

The $\zeta-$KS diagram is shown in Fig. \vref{fig:KS-diagram}. The
$\zeta-$KS plane is conformally flat. The causal structure of the
regular KS diagram remains intact in the $\zeta-$KS diagram. \textcolor{black}{In
the interior, null and timelike geodesics will eventually hit the
origin; namely, no physical objects can escape the interior. In the
exterior, outgoing light paths can escape to infinity, whereas incoming
light paths must fall into the interior. See Sec. \ref{subsec:KS-diagram}.
\vskip4pt}

Yet there emerges a very surprising feature in the $\zeta-$KS diagram.
Sandwiching between the four known quadrants (I)--(IV) is an ``virtual''
domain which cannot be mapped to any region in the original manifold
specified by $(t,\rho,\theta,\phi)$. Transits of physical objects
from the exterior into the interior must bypass this ``gulf'' unaffected,
at least at the classical level.\vskip4pt

Given that the $\zeta-$KS diagram is the maximal extension of metric
\eqref{eq:metric-summary-1}--\eqref{eq:metric-summary-2}, the ``gulf''
that emerges is a tantalizing aspect, deserving further investigation.
We put forth a conjecture that the ``virtual gulf'' could be accounted
for by embedding the $\zeta-$tortoise coordinate into the complex
plane. See our Conjecture \ref{conj:my-conjecture}.

\subsubsection{Questioning the validity of techniques based on series expansions
around the interior-exterior boundary}

The \emph{non-analyticity} of the \emph{special} Buchdahl-inspired
metric across the interior-exterior boundary is self-evident in the
singularities of the Kretschmann scalar, the anomalous properties
of the surface area of the interior-exterior boundary, and the appearance
of a ``virtual gulf'' in the $\zeta-$KS plane. This metric therefore
\emph{cannot} be attained by any technique that is based on an analytic
perturbative expansion around the interior-exterior boundary.\vskip4pt

In a larger context, for the \emph{full} quadratic gravity, viz. $\gamma\,\mathcal{R}+\beta\,\mathcal{R}^{2}-\alpha\,\mathcal{C}^{\mu\nu\rho\sigma}\mathcal{C}_{\mu\nu\rho\sigma}$,
as the generalized Lichnerowicz theorem has been evaded, one must
restore the $\mathcal{R}^{2}$ term, namely, permitting $\beta\neq0$;
see \citep{Nguyen-2022-extension}. Solutions with non-analytic behaviors
across the interior-exterior boundary should be possible. At the very
least, the limit of $\alpha=\gamma=0$ must recover the \emph{special}
Buchdahl-inspired metric together with its anomalies. The L\"u-Perkins-Pope-Stelle
ansatz made in \citep{Lu-2015-a,Lu-2015-b} would need augmenting
with non-analytic built-ins in order to find these solutions in the
full quadratic action. See our companion paper for discussions \citep{Nguyen-2022-extension}.

\subsubsection{\label{subsec:Non-Schwarzschild-structures}Non-Schwarzschild structures
in pure $\mathcal{R}^{2}$ spacetime}

The divergence of the Kretschmann invariant at the interior-exterior
boundary, $\rho=r_{\text{s}}$, for $k\neq0$ and $k\neq-r_{\text{s}}$
signals the formation of a naked singularity or a wormhole. Given
that pure $\mathcal{R}^{2}$ gravity is equivalent to a scalar-tensor
theory, it would be natural to consider the special Buchdahl-inspired
metric in conjunction with exact solutions in Brans-Dicke gravity,
viz. the Brans and Campanelli-Lousto solutions which are known to
possess naked singularities or wormholes, depending on the value of
the Brans-Dicke parameter \citep{Agnese-1995,Vanzo-2012,Brans-1962,Campanelli-1993,Faraoni-2016}.
The no-hair theorem first proved by Hawking \citep{Hawking-1972-BD}
and later generalized by Sotiriou and Faraoni \citep{SotiriouFaraoni-2012}
for scalar-tensor gravity should also be taken in account. We plan
to investigate this direction in future research.

\textcolor{black}{\vskip8pt}
\begin{center}
-----------------$\infty$-----------------\vskip8pt
\par\end{center}

What is surprising is that pure $\mathcal{R}^{2}$ gravity is a parsimonious
theory, containing only one single term in the action \footnote{Besides its parsimony, virtues of this theory are in being ghost-free
and scale invariant \citep{Stelle-1977,Stelle-1978,Luest-2015-fluxes}.}. It does not involve exogenous terms, torsion, non-metricity, metric-affine
hybrid, or non-locality \citep{Clifton-2011,Sotiriou-2008,deFelice-2010,Capozziello-2011}.
It operates within the vanilla local metric-based formalism.\linebreak
Yet, despite its simplicity, it already produces novel behaviors,
reported herein, that are yet encountered in the Einstein-Hilbert
theory. Moreover, pure $\mathcal{R}^{2}$ gravity admits the Buchdahl-inspired
vacua with \emph{non-constant} scalar curvature, per Eq. \eqref{eq:B-metric-4}.
The asymptotic scalar curvature $4\Lambda$ and the Buchdahl parameter
$k$ are two \emph{endogenous} degrees of freedom that are only accessible
in a fourth-order theory, as opposed to a second-order theory such
as the Einstein-Hilbert action.\vskip6pt

It is the Buchdahl parameter $k$ that enriches $\mathcal{R}^{2}$
gravity with phenomenology which transcends the Einstein-Hilbert paradigm.

\section{Closing words}

In this second installment of our three-paper ``Beyond Schwarzschild--de
Sitter spacetimes'' series \citep{Nguyen-2022-Buchdahl,Nguyen-2022-extension},
we reported an exact closed analytical solution that serves as a bona
fide enlargement of the Schwarzschild solution. It encloses the Schwarzschild
spacetime as a limiting case (when the Buchdahl parameter $k$ is
sent to zero). We achieved this result by advancing an unfinished
program in search of pure $\mathcal{R}^{2}$ vacua, a program that
was originated but ``forsaken'' by Buchdahl circa 1962, and largely
``forgotten'' by the gravitation research community in the past
sixty years \citep{Buchdahl-1962}. Novel intriguing theoretical properties
of $\mathcal{R}^{2}$ spacetime structures are uncovered and reported
herein, suggesting that the Buchdahl-inspired spacetimes may fall
outside of the Einstein-Hilbert paradigm. They may well belong to
a separate Buchdahl-inspired framework, warranting further explorations.
\begin{acknowledgments}
I thank the anonymous referee for their important comments in improving
the paper and stimulating further developments, especially regarding
the non-Schwarzschild $\mathcal{R}^{2}$ structures. I thank Dieter
L\"ust for his encouragement during the development of this research.
The anonymous referee of my previous paper \citep{Nguyen-2022-Buchdahl}
motivated me to strengthen the capacity of my work in evading the
generalized Lichnerowicz theorem \citep{Nelson-2010,Lu-2015-a,Lu-2015-b,Luest-2015-backholes}.
The valuable help and technical insights from Richard Shurtleff are
acknowledged. I thank Tiberiu Harko for his supports, Sergei Odintsov
and Timothy Clifton for their comments.
\end{acknowledgments}

\begin{center}
-----------------$\infty$-----------------
\par\end{center}

\appendix

\section{\label{sec:Case-rS0}$\ $The case of $\boldsymbol{r_{\text{s}}=0}$}

From Lemma \ref{lem:lemma-1} and Corollary \ref{cor:first-transf},
we have
\begin{align}
q_{\pm} & =\frac{\sqrt{3}}{2}\left|k\right|\\
r & =\left|q^{2}-\frac{3}{4}k^{2}\right|^{\frac{1}{2}}\\
p & =\text{sgn}\Bigl(q^{2}-\frac{3}{4}k^{2}\Bigr)\frac{\left|q^{2}-\frac{3}{4}k^{2}\right|^{\frac{1}{2}}}{q}\\
\frac{p\,q}{r} & =\text{sgn}\Bigl(q^{2}-\frac{3}{4}k^{2}\Bigr)\\
e^{\omega} & =\left|\frac{q-\frac{\sqrt{3}}{2}\left|k\right|}{q+\frac{\sqrt{3}}{2}\left|k\right|}\right|^{\frac{2}{\sqrt{3}}\text{sgn}\left(k\right)}
\end{align}
The metric is thus\small
\begin{align}
ds^{2} & =\left|\frac{q-\frac{\sqrt{3}}{2}\left|k\right|}{q+\frac{\sqrt{3}}{2}\left|k\right|}\right|^{\frac{2}{\sqrt{3}}\text{sgn}\left(k\right)}\times\nonumber \\
 & \ \left\{ \text{sgn}\Bigl(q^{2}-\frac{3}{4}k^{2}\Bigr)\left[-dt^{2}+dq^{2}\right]+\left|q^{2}-\frac{3}{4}k^{2}\right|d\Omega^{2}\right\} 
\end{align}
\normalsize

\section{\label{sec:Gaussian-hypergeometric-function}$\ $Gaussian hypergeometric
function}

The Gaussian hypergeometric function involved in the $\zeta-$tortoise
coordinate, $\,_{2}F_{1}(a,b;c;z)$ in terms of series
\begin{equation}
\,_{2}F_{1}(a,b;c;z)=1+\frac{ab}{c.1!}z+\frac{a(a+1)b(b+1)}{c(c+1).2!}z^{2}+\dots
\end{equation}
Generally speaking, this series converges in the unit circle $\left|z\right|<1$.
For the $\zeta-$tortoise coordinate (modulo an additive constant)\vskip-10pt

\small
\begin{align}
\rho^{*} & =\frac{\zeta^{2}r_{\text{s}}}{\zeta-1}\,z^{1-\frac{1}{\zeta}}\,_{2}F_{1}\left(2,1-\frac{1}{\zeta};2-\frac{1}{\zeta};\text{sgn}\Bigl(1-\frac{r_{\text{s}}}{\rho}\Bigr)z\right)
\end{align}
\normalsize in which $z:=\left|1-\frac{r_{\text{s}}}{\rho}\right|^{\zeta}$,
or equivalently, $\rho>r_{\text{s}}/2$ (note that $\zeta:=\sqrt{1+3\tilde{k}^{2}}>1$
for $\tilde{k}\neq0$).

For $0<\rho<r_{\text{s}}/2$, in order to continue using a hypergeometric
function defined via a series, we would need to ``invert'' the variable
$z$. Recall the ODE for $\rho^{*}$ (for $\rho<r_{\text{s}}$):
\begin{equation}
d\rho^{*}=+\zeta r_{\text{s}}\frac{z^{-1/\zeta}}{(1+z)^{2}}dz
\end{equation}
Let us substitute $z:=y^{-1}$, then
\begin{equation}
d\rho^{*}=-\zeta\frac{y^{1/\zeta}}{(1+y)^{2}}dy
\end{equation}
accepting the solution (modulo an additive constant)
\begin{align}
\rho^{*} & =-\frac{\zeta{}^{2}r_{\text{s}}}{\zeta+1}\,z^{-1-\frac{1}{\zeta}}\,_{2}F_{1}\left(2,1+\frac{1}{\zeta};2+\frac{1}{\zeta};-z^{-1}\right)
\end{align}
which converges for $\rho<r_{\text{s}}$. Note that its is nothing
but the original solution with $\zeta$ replaced by $-\zeta$ (including
the $\zeta$ in the definition of $z$).

\section{\label{sec:Limit-k-0}$\ $The $\boldsymbol{k\rightarrow0}$ limit
of the $\boldsymbol{\zeta-}$tortoise coordinate}

In the limit of $k\rightarrow0$, viz. $\zeta\rightarrow1$:\small
\begin{equation}
\left|z\right|^{1-\frac{1}{\zeta}}=1+\ln\left|z\right|\Bigl(1-\frac{1}{\zeta}\Bigr)+\mathcal{O}\left(\Bigl(1-\frac{1}{\zeta}\Bigr)^{2}\right)
\end{equation}
and
\begin{align}
 & \frac{\zeta}{\zeta-1}\,_{2}F_{1}\left(2,1-\frac{1}{\zeta};2-\frac{1}{\zeta};z\right)\nonumber \\
 & \ \ \ =\frac{1}{1-\frac{1}{\zeta}}\sum_{n=0}^{\infty}\frac{(n+1)\Bigl(1-\frac{1}{\zeta}\Bigr)}{n+1-\frac{1}{\zeta}}z^{n}\\
 & \ \ \ =\frac{1}{1-\frac{1}{\zeta}}+\sum_{n=1}^{\infty}\frac{n+1}{n+1-\frac{1}{\zeta}}z^{n}\\
 & \ \ \ =\frac{1}{\zeta-1}+\left(1+\sum_{n=1}^{\infty}z^{n}\right)+\frac{1}{\zeta}\sum_{n=1}^{\infty}\frac{z^{n}}{n+1-\frac{1}{\zeta}}\\
 & \ \ \ =\frac{1}{\zeta-1}+\frac{1}{1-z}+\frac{1}{\zeta}\sum_{n=1}^{\infty}\left[\frac{z^{n}}{n}+\mathcal{O}\Bigl(1-\frac{1}{\zeta}\Bigr)\right]\\
 & \ \ \ =\frac{1}{\zeta-1}+\frac{1}{1-z}-\frac{1}{\zeta}\ln\left|1-z\right|+\mathcal{O}\Bigl(1-\frac{1}{\zeta}\Bigr)
\end{align}
\normalsize Eq. \eqref{eq:tortoise-1} gives \small
\begin{align}
\frac{\rho^{*}}{r_{\text{s}}} & =\frac{\zeta}{\zeta-1}+\frac{\zeta}{1-z}-\ln\left|1-z\right|+\ln\left|z\right|+\mathcal{O}\Bigl(1-\frac{1}{\zeta}\Bigr)
\end{align}
\normalsize Note that $\rho^{*}$ was determined up to an additional
constant. In the limit $\zeta\rightarrow1$, we are thus left with
\small
\begin{equation}
\frac{\rho^{*}}{r_{\text{s}}}=\frac{1}{1-z}+\ln\left|1+\frac{1}{z-1}\right|
\end{equation}
\normalsize Taking into account Eq. \eqref{eq:def-z}, viz. $z=\text{sgn}\Bigl(1-\frac{r_{\text{s}}}{\rho}\Bigr)\times\left|1-\frac{r_{\text{s}}}{\rho}\right|^{\zeta}=1-\frac{r_{\text{s}}}{\rho}$
for $\zeta=1$, we finally have
\begin{equation}
\rho^{*}=\rho+r_{\text{s}}\ln\left|\frac{\rho-r_{\text{s}}}{r_{\text{s}}}\right|
\end{equation}
in agreement with the usual tortoise coordinate.

\end{document}